\documentclass[preprint,authoryear,3p,12pt]{elsarticle}

\usepackage{amssymb}
\usepackage{amsthm}

\usepackage{mathrsfs}
\usepackage{float}
\usepackage{eufrak}

\journal{}

\newtheorem{theorem}{Theorem}[section]
\newtheorem{prop}[theorem]{Proposition}
\newtheorem{cor}[theorem]{Corollary}
\newtheorem{lemma}[theorem]{Lemma}
\newtheorem{remark}[theorem]{Remark}
\newtheorem{define}[theorem]{Definition}
\newtheorem{example}[theorem]{Example}

\newenvironment{thm31}
{{\vspace*{10pt}} \par \noindent \textbf{Theorem 4.3 (Correctness
Theorem).}~\itshape} {\par {\vspace*{10pt}}}

\newenvironment{prop36}
{{\vspace*{10pt}} \par \noindent \textbf{Proposition 4.6 (first
kind).}~\itshape} {\par {\vspace*{10pt}}}

\newenvironment{prop38}
{{\vspace*{10pt}} \par \noindent \textbf{Proposition 4.9
(comparable).}~\itshape} {\par {\vspace*{10pt}}}

\newenvironment{prop39}
{{\vspace*{10pt}} \par \noindent \textbf{Proposition 4.10
(rewritable).}~\itshape} {\par {\vspace*{10pt}}}

\newenvironment{lm310}
{{\vspace*{10pt}} \par \noindent \textbf{Lemma 4.11 (key
lemma).}~\itshape} {\par {\vspace*{10pt}}}

\newenvironment{prop311}
{{\vspace*{10pt}} \par \noindent \textbf{Proposition 4.12
(left).}~\itshape} {\par {\vspace*{10pt}}}

\newenvironment{prop312}
{{\vspace*{10pt}} \par \noindent \textbf{Proposition 4.13
(right).}~\itshape} {\par {\vspace*{10pt}}}

\floatstyle{ruled}
\newfloat{algorithm}{h}{loa}
\floatname{algorithm}{Algorithm}
\newcommand{\Input}[1]
  {\noindent\begin{tabular}{@{}p{1.8cm}@{}p{13.2cm}@{}}
   {\bf Input: }&#1 \end{tabular}}
\newcommand{\Output}[1]
  {\noindent\begin{tabular}{@{}p{1.8cm}@{}p{13.2cm}@{}}
   {\bf Output: }&#1 \end{tabular}}

\floatstyle{ruled}
\newfloat{function}{h}{loa}
\floatname{function}{Function}

\floatstyle{ruled}
\newfloat{criterion}{!h}{loa}
\floatname{criterion}{Criterion}

\def\F{{\mathbb{F}}}
\def\N{{\mathbb{N}}}
\def\K{{\rm K}}

\def\m{{\rm m}}
\def\lm{{\rm lm}}
\def\lpp{{\rm lpp}}
\def\lc{{\rm lc}}

\def\fb{{\bf f}}
\def\eb{{\bf e}}
\def\gb{{\bf g}}
\def\hb{{\bf h}}
\def\lcm{{\rm lcm}}
\def\gcd{{\rm gcd}}

\def\deg{{\rm deg}}

\def\max{{\rm max}}

\def\sp{{\rm spoly}}

\def\Q{{\mathbb{Q}}}
\def\gr{{Gr\"obner }}
\def\x{{x_1,\cdots,x_n}}

\def\sus{{\succ\,}}
\def\prl{{\lhd\,}}
\def\sul{{\rhd\,}}
\def\prc{{\lhd\,}}

\def\sull{{\rhd'\,}}
\def\s{{\rm{Sign}}}
\def\p{{\rm{Poly}}}
\def\n{{\rm{Num}}}

\def\lif{{\bf if \,}}
\def\lthen{{\bf then \,}}
\def\lelse{{\bf else \,}}
\def\lendif{{\bf end if\,}}

\def\lwhile{{\bf while \,}}

\def\lendwhile{{\bf end while\,}}
\def\ldo{{\bf do \,}}

\def\lend{{\bf end \,}}
\def\lreturn{{\bf return \,}}
\def\lla{{\longleftarrow}}

\def\lbegin{{\bf begin}}

\newcommand{\SPC}{\hspace*{15pt}}
\newcommand{\ignore}[1]{}

\begin{document}

\begin{frontmatter}



\title{A New Proof for the Correctness of F5 (F5-Like) Algorithm}


\author{Yao Sun and Dingkang Wang\fnref{label1}}

\fntext[label1]{The authors are supported by NSFC 10971217,
10771206 60821002/F02.}

\address{Key Laboratory of Mathematics Mechanization, Academy of Mathematics and Systems Science, CAS, Beijing 100190,  China}

\ead{sunyao@amss.ac.cn, dwang@mmrc.iss.ac.cn}

\begin{abstract}
The famous F5 algorithm for computing \gr basis was presented by
Faug\`ere in 2002 without complete proofs for its correctness. The
current authors have simplified the original F5 algorithm into an F5
algorithm in Buchberger's style (F5B algorithm), which is equivalent
to original F5 algorithm and may deduce some F5-like versions. In
this paper, the F5B algorithm is briefly revisited and a new
complete proof for the correctness of F5B algorithm is proposed.
This new proof is not limited to homogeneous systems and does not
depend on the strategy of selecting critical pairs (i.e. the
strategy deciding which critical pair is computed first) such that
any strategy could be utilized in F5B (F5) algorithm. From this new
proof, we find that the special reduction procedure (F5-reduction)
is the key of F5 algorithm, so maintaining this special reduction,
various variation algorithms become available. A natural variation
of F5 algorithm, which transforms original F5 algorithm to a
non-incremental algorithm, is presented and proved in this paper as
well. This natural variation has been implemented over the Boolean
ring. The two revised criteria in this natural variation are also
able to reject almost all unnecessary computations and few
polynomials reduce to 0 in most examples.
\end{abstract}


\end{frontmatter}



\section{Introduction} \label{sec-introduction}

Solving systems of polynomial equations is a basic problem in
computer algebra, through which many practical problems can be
solved easily. Among all the methods for solving polynomial systems,
the \gr basis method is one of the most efficient approaches. After
the conception of \gr basis is proposed in 1965
\citep{Buchberger65}, many algorithms have been presented for
computing \gr basis, including \citep{Lazard83, GebMol86, Gio91,
Mora92, Fau99, Fau02}. Currently, F5 algorithm is one of the most
efficient algorithms.

After the F5 algorithm is proposed, many researches have been done.
For example, Bardet et al. study the complexity of this algorithm in
\citep{Bardet03}. Faug\`ere and Ars use the F5 algorithm to attack
multivariable systems in \citep{Fau03}. Stegers revisits F5
algorithm in his master thesis \citep{Stegers05}. Eder discusses the
two criteria of F5 algorithm in \citep{Eder08} and proposes a
variation of F5 algorithm \citep{Eder09}. Ars and Hashemi present
two variation of criteria in \citep{Ars09}. Recently, Gao et al.
give a new incremental algorithm in \citep{Gao09}. The current
authors discuss the F5 algorithm over boolean ring and present a
branch F5 algorithm in \citep{SunWang09a, SunWang09b}. We also
discuss the F5 algorithm in Buchberger's style in \citep{SunWang10}.

Currently, available proofs for the correctness of F5 algorithm can
be found from \citep{Fau02, Stegers05, Eder08, Eder09}. However,
these proofs are somewhat not complete, particularly for
non-homogeneous systems.

The main purpose of current paper is to present a new complete proof
for the correctness of F5 (F5-like) algorithm. As we have shown in
\citep{SunWang10} that the F5 algorithm in Buchberger's style (F5B
algorithm) is equivalent to the original F5 algorithm in
\citep{Fau02} and may deduces various F5-like algorithms, therefore,
we will focus on proving the correctness of F5B algorithm in this
paper. The proposed new proof is {\bf not} limited to homogeneous
systems and does {\bf not} depend on the strategy of selecting
critical pairs (s-pairs), so the correctness of all versions of F5
algorithm mentioned in \citep{SunWang10} can be proved at the same
time. After a slight modification, the correctness of the variation
of F5 algorithm in \citep{Ars09}, which is quite similar as the
natural variation in this paper, can also be proved.

Meanwhile, according to the new proposed proof, we find that the key
of F5 (F5-like) algorithm is the special reduction procedure, which
ensures the correctness of both criteria in F5 algorithm. Thus,
maintaining this special reduction procedure, many variations of F5
algorithm become available. We propose and prove a natural variation
of F5 algorithm after the main proofs. This variation algorithm
avoids computing \gr basis incrementally such that the \gr bases for
subsets of input polynomials are not necessarily computed. Besides,
the two revised criteria in this variation are also able to reject
almost all unnecessary reductions as shown in the experimental data.

This paper is organized as follows. We revisit the F5 algorithm in
Buchberger's style (F5B algorithm) in Section 3 after introducing
basic notations in Section 2. The complete proof for the correctness
of F5B algorithm is presented in Section 4. The key of F5 algorithm
and the natural variation algorithm are discussed in Section 5. This
paper is concluded in Section 6.

\section{Basic Notations}

Let $\K$ be a field and $\K[X]=\K[\x]$ a polynomial ring with
coefficients in $\K$. Let $\N$ be the set of non-negative integers
and $PP(X)$ the set of power products of $\{x_1,\cdots,x_n\}$,
i.e. $PP(X):=\{x^\alpha  \mid x^\alpha=x_1^{\alpha_1}\cdots
x_n^{\alpha_n}, \alpha_i \in \N, i=1,\cdots,n\}$.

Let $\prec$ be an admissible order defined over $PP(X)$. Given
$t=x^{\alpha} \in PP(X)$, the degree of $t$ is defined as $\deg(t):=
|{\alpha}| =\sum_{i=1}^n\alpha_i$. For a polynomial $0\not=f\in
\K[x_1,\cdots,x_n]$, we have $f=\sum c_{\alpha}x^{\alpha}$. The
degree of $f$ is defined as $\deg(f):=\max\{| \alpha |,
c_{\alpha}\not=0\}$ and the leading power product of $f$ is
$\lpp(f):=\max_\prec\{x^{\alpha}, c_{\alpha}\not=0\}$. If
$\lpp(f)=x^{\alpha}$, then the leading coefficient and leading
monomial of $f$ are defined to be $\lc(f):=c_{\alpha}$ and
$\lm(f):=c_{\alpha}x^{\alpha}$ respectively.

\section{The F5 Algorithm in Buchberger's Style} \label{sec-F5B}

In brief, F5 algorithm introduces a special reduction
(F5-reduction) and provides two new criteria (Syzygy Criterion
\footnote{Also called F5 Criterion in some papers. To avoid
confusion with the name of F5 algorithm, we call it Syzygy
Criterion in this paper.} and Rewritten Criterion) to avoid
unnecessary reductions.

In this section, we give the definitions of signatures and labeled
polynomials first, and then describe the Syzygy Criterion and
Rewritten Criterion as well as the special reduction, F5-reduction.
At last, we present the F5 algorithm in Buchberger's Style (F5B
algorithm) as discussed in \citep{SunWang10}.

As a preparation for the main proofs, an important auxiliary concept
is introduced. That is the {\em numbers} of labeled polynomials,
which reflect the order of when labeled polynomials are generated.
This auxiliary concept simplifies the description of the Rewritten
Criterion and benefits for the main proofs. For more details about
the F5B algorithm, please see \citep{SunWang10}.

\subsection{Signature and Labeled Polynomial} \label{subsec-signature}

Consider a polynomial system $\{f_1, \cdots, f_m\}\subset \K[X]$ and
denote $(f_1, \cdots, f_m)$ to be a polynomial $m$-tuple in
$(\K[X])^m$. We call the $f_i$'s initial polynomials, as they are
initial generators of ideal $\langle f_1, \cdots, f_m \rangle\subset
\K[X]$.

Let $\eb_i$ be the canonical $i$-th unit vector
in $(\K[X])^m$, i.e. the $i$-th element of $\eb_i$ is $1$, while
the others are $0$ . Consider the homomorphism map $\sigma$ over
the free module $(\K[X])^m$:
$$\sigma: (\K[X])^m \longrightarrow \langle f_1, \cdots, f_m
\rangle,$$
$$(g_1, \cdots, g_m) \longmapsto g_1f_1+\cdots+g_mf_m.$$ Then $\sigma(\eb_i)=f_i$. More
generally, if $\gb=g_1\eb_1+\cdots+ g_m\eb_m$, where $g_i\in
\K[X]$ for $1\le i\le m$, then $\sigma(\gb)=g_1f_1+\cdots+g_mf_m$.

The admissible order $\prec$ on $PP(X)$ extends to the free module
$(\K[X])^m$ naturally in a POT (position over term)
fashion\footnote{This order of signatures is imported from
\citep{Fau02}. We will introduce another order of signatures to
deduce the natural variation of F5 algorithm after the main
proofs.}:
$$x^\alpha\eb_i \prec x^\beta\eb_j \mbox{ (or } x^\beta\eb_j \succ x^\alpha\eb_i)\ \ \mbox{  iff  }
\left\{\begin{array}{l} i > j, \\ \mbox{ or } \\ i = j \mbox{ and
} x^\alpha \prec x^\beta.
\end{array}\right.
$$ Thus we have $\eb_m \prec \eb_{m-1} \prec \cdots \prec
\eb_1$.

With the admissible order on $(\K[X])^m$, we can define the leading
power product, leading coefficient and leading monomial of a
$m$-tuple vector $\gb\in (\K[X])^m$ in a similarly way. For example,
let $\gb=(2x^2+y^2, 3xy)\in (\Q[x, y])^2$ or equivalently
$\gb=(2x^2+y^2)\eb_1+3xy\eb_2$. According to the Lex order $\prec$
on $PP(x, y)$ where $x \succ y$, we have $\lpp(\gb)=x^2\eb_1$,
$\lc(\gb)=2$ and $\lm(\gb)=2x^2\eb_1$.

Now, we give the mathematical definition of signatures.

\begin{define}[signature] \label{df-signature}
Let $\gb\in (\K[X])^m$ be an $m$-tuple vector. If polynomial
$g=\sigma(\gb)\in \langle f_1, \cdots, f_m\rangle \subset \K[X]$,
then the leading power product $\lpp(\gb)$ is defined to be a
signature of $g$.
\end{define}

Consider a simple system $\{f_1=x^2+2y, f_2=xy-z\} \subset \Q[x, y]$
with the Graded Reverse Lex Order ($x\succ y$). The s-polynomial of
$f_1$ and $f_2$ is $yf_1-xf_2=2y^2+xz$. According to the above
definition, $y\eb_1$ is a signature of the polynomial $2y^2+xz$, as
$\sigma(y\eb_1-x\eb_2)=2y^2+xz$ and $\lpp(y\eb_1-x\eb_2)=y\eb_1$.

Now we are able to assign a signature to each polynomial $g\in
\langle f_1, \cdots, f_m\rangle$. To tighten up the relation
between a polynomial and its signature, we integrate them together
and call it {\em labeled polynomial}.

\begin{define}[labeled polynomial]
Let $g\in \langle f_1, \cdots, f_m\rangle$ be a polynomial. If
$x^\alpha\eb_i$ is a signature of $g$, then ${\mathcal
G}=(x^\alpha\eb_i, g, k)$ is defined to be a labeled polynomial of
$g$, where $k\in \N$ reflects the order of when this labeled
polynomial is generated.

For convenience, denote the signature, polynomial and number of
the labeled polynomial ${\mathcal G}$ by $\s({\mathcal
G}):=x^\alpha\eb_i$, $\p({\mathcal G}):=g$ and $\n({\mathcal
G}):=k$. Besides, the leading power product and leading monomial
of ${\mathcal G}$ are denoted as: $\lpp({\mathcal G}):=\lpp(g)$
and $\lm({\mathcal G}):=\lm(g)$ respectively.
\end{define}

The {\em number} of labeled polynomial is an important auxiliary
concept for the main proofs. It is designated by the algorithm and
reflects the order of when the labeled polynomials are generated.
The meaning of {\em number} will be much clearer after the F5B
algorithm is presented.

Remark that a polynomial in the ideal $\langle f_1, \cdots,
f_m\rangle$ may have several different signatures, but during the
computations, the {\em signature} and {\em number} of each
polynomial are uniquely determined by the algorithm.

Therefore, in the above simple example $\{f_1=x^2+y, f_2=xy-z\}
\subset \Q[x, y]$, the labeled polynomials corresponding to $f_1$
and $f_2$ are $(\eb_1, f_1, 1)$ and $(\eb_2, f_2, 2)$
respectively. For the s-polynomial $yf_1-xf_2=2y^2+xz$ of $f_1$
and $f_2$, its labeled polynomial is $(y\eb_1, 2y^2+xz,
3)$.\footnote{The {\em number} of this s-polynomial is designated
by the algorithm} Then the {\em signature, polynomial} and {\em
number} of ${\mathcal G}=(y\eb_1, 2y^2+xz, 3)$ are $\s({\mathcal
G})=y\eb_1$, $\p({\mathcal G})=2y^2+xz$ and $\n({\mathcal G})=3$
respectively. We also have $\lpp({\mathcal G})=y^2$ and
$\lm({\mathcal G})=2y^2$. Notice that the {\em numbers} of labeled
polynomials $(\eb_1, f_1, 1)$ and $(\eb_2, f_2, 2)$ are $1$ and
$2$, both of which are smaller than the {\em number} $\n({\mathcal
G})=3$. This indicates the labeled polynomial ${\mathcal
G}=(y\eb_1, 2y^2+xz, 3)$ is generated later than labeled
polynomials $(\eb_1, f_1, 1)$ and $(\eb_2, f_2, 2)$.

Now we introduce two notations about signatures and labeled
polynomials. Define $S(X)$ to be the set of signatures, i.e.
$S(X):=\{x^\alpha \eb_i \mid  x^\alpha\in PP(X), 1\le i\le
m\}\subset R^m$, and $L(X)$ to be the set of labeled polynomials,
i.e. $L(X):=\{(x^\alpha \eb_i, g, k_g) \mid x^\alpha\eb_i \in S(X)
\mbox{ is a signature of } g\in \K[X]\}$. In the rest of current
paper, we use the flourish, such as ${\mathcal F}, {\mathcal G},
{\mathcal H}$, to represent labeled polynomials, while the
lowercase, such as $f, g, h$, stand for polynomials in $\K[X]$. The
boldface, $\fb, \gb, \hb$, refer to the elements in free module
$(\K[X])^m$.

In F5 algorithm, labeled polynomials are the basic elements in
computation instead of polynomials in $\K[X]$.
Suppose $\fb, \gb \in (K[X])^m$ such that $\sigma(\fb)=f$ and
$\sigma(\gb)=g$. Then ${\mathcal F}=(\lpp(\fb), f, k_f), {\mathcal
G}=(\lpp(\gb), g, k_g)\in L[X]$ are labeled polynomials. Assume
$cx^\gamma$ is a non-zero monomial. Then
\begin{enumerate}

\item[$\bullet$] $cx^\gamma{\mathcal F}=(x^\gamma\lpp(\fb),
cx^\gamma f, k_f)$, as $\sigma(cx^\gamma\fb)=cx^\gamma f$.

\item[$\bullet$] ${\mathcal F}+{\mathcal
G}=(\max_{\prec}\{\lpp(\fb), \lpp(\gb)\}, f+g, k_{f, g})$, as
$\sigma(\fb+\gb)=f+g$, where $k_{f, g}=k_f$ or $k_g$ corresponding
to the maximal one of $\{\lpp(\fb), \lpp(\gb)\}$.

\end{enumerate}

Unlike polynomials in $\K[X]$, labeled polynomials in $L[X]$ can
compare in following way:
$$(x^\alpha\eb_i, f, k_f) \prl (x^\beta\eb_j, g, k_g)
\mbox{ (or }(x^\beta\eb_j, g, k_g) \sul (x^\alpha\eb_i, f, k_f))\
\
\mbox{  iff  } \left\{\begin{array}{l} x^\alpha\eb_i \prec x^\beta\eb_j, \\
\mbox{ or } \\ x^\alpha\eb_i = x^\beta\eb_j \mbox{ and } k_f >
k_g.
\end{array}\right.
$$
Particularly, denote $(x^\alpha\eb_i, f, k_f) \bowtie
(x^\beta\eb_j, g, k_g)$, if $x^\alpha\eb_i = x^\beta\eb_j$ and
$k_f=k_g$. Remark that in this case, the polynomial $f$ may not
equal to $g$.

In the simple example $\{f_1=x^2+y, f_2=xy-z\} \subset \Q[x, y]$.
We have $(\eb_2, f_2, 2) \prl (\eb_1, f_1, 1)$, since $\eb_2 \prec
\eb_1$. For the s-polynomial $yf_1-xf_2=2y^2+xz$ of $f_1$ and
$f_2$, its labeled polynomial is $(y\eb_1, 2y^2+xz, 3)$. Notice
that $y(\eb_1, f_1, 1)=(y\eb_1, yf_1, 1)$. So we also have
$(y\eb_1, 2y^2+xz, 3) \prl (y\eb_1, yf_1, 1)$ due to the {\em
numbers} of these two labeled polynomials.

The critical pair (s-pair) of labeled polynomials is defined in a
similar way as well. For labeled polynomials ${\mathcal F},
{\mathcal G}\in L[X]$, we say $[{\mathcal F}, {\mathcal G}]:=(u,
{\mathcal F}, v, {\mathcal G})$ is the {\em critical pair} of
${\mathcal F}$ and ${\mathcal G}$, if $u, v$ are monomials in $X$
such that $u\lm({\mathcal F})=v\lm({\mathcal
G})=\lcm(\lpp({\mathcal F}), \lpp({\mathcal G}))$ and $u{\mathcal
F} \sul v{\mathcal G}$. Besides, the {\em s-polynomial} of
$[{\mathcal F}, {\mathcal G}]=(u, {\mathcal F}, v, {\mathcal G})$
is denoted as $\sp({\mathcal F}, {\mathcal G})=u{\mathcal
F}-v{\mathcal G}$.

Remark that labeled polynomials in the critical pair $[{\mathcal F},
{\mathcal G}]=(u, {\mathcal F}, v, {\mathcal G})$ is ordered by
$u{\mathcal F} \sul v{\mathcal G}$. Moreover, critical pairs can
compare with each other in the following way:
$$(u, {\mathcal F}, v, {\mathcal G}) \prl (r, {\mathcal P}, t, {\mathcal Q})
\mbox{ (or }(r, {\mathcal P}, t, {\mathcal Q}) \sul (u, {\mathcal
F}, v, {\mathcal G}))\ \
\mbox{  iff  } \left\{\begin{array}{l} u{\mathcal F} \prl r{\mathcal P}, \\
\mbox{ or } \\ u{\mathcal F} \bowtie r{\mathcal P} \mbox{ and }
v{\mathcal G} \prl t{\mathcal Q}.
\end{array}\right.
$$

\subsection{Syzygy Criterion and Rewritten Criterion}

First, we describe the Syzygy Criterion. We begin by the following
definition.

\begin{define}[Comparable]
Let ${\mathcal F}=(x^\alpha\eb_i, f, k_f)\in L[X]$ be a labeled
polynomial, $cx^\gamma$ a non-zero monomial in $X$ and $B\subset
L[X]$ a set of labeled polynomials. The labeled polynomial
$cx^\gamma{\mathcal F}=(x^{\gamma+\alpha}\eb_i, cx^\gamma f, k_f)$
is said to be comparable by $B$, if there exists a labeled
polynomial ${\mathcal G}=(x^\beta\eb_j, g, k_g)\in B$ such that:
\begin{enumerate}

\item $\lpp(g) \mid x^{\gamma+\alpha}$, and

\item $\eb_i \succ \eb_j$, i.e. $i<j$.
\end{enumerate}
\end{define}

Then the Syzygy Criterion is described as follow.

\begin{criterion} \caption{\bf --- Syzygy Criterion}
  \smallskip
  \medskip

  \noindent
Let $[{\mathcal F}, {\mathcal G}]:=(u, {\mathcal F}, v, {\mathcal
G})$ be the critical pair of ${\mathcal F}$ and ${\mathcal G}$,
where $u, v$ are monomials in $X$ such that $u\lm({\mathcal
F})=v\lm({\mathcal G})=\lcm(\lpp({\mathcal F}), \lpp({\mathcal
G}))$ and $u{\mathcal F} \sul v{\mathcal G}$. And $B\subset L[X]$
is a set of labeled polynomials. If either $u{\mathcal F}$ or
$v{\mathcal G}$ is {\bf comparable} by $B$, then the critical pair
$[{\mathcal F}, {\mathcal G}]$ meets the Syzygy Criterion.
  \medskip
\end{criterion}

Next, we describe the Rewritten Criterion. Again we start with a
definition.

\begin{define}[Rewritable]
Let ${\mathcal F}=(x^\alpha\eb_i, f, k_f)\in L[X]$ be a labeled
polynomial, $cx^\gamma$ a non-zero monomial in $X$ and $B\subset
L[X]$ a set of labeled polynomials. The labeled polynomial
$cx^\gamma{\mathcal F}=(x^{\gamma+\alpha}\eb_i, cx^\gamma f, k_f)$
is said to be rewritable by $B$, if there exists a labeled
polynomial ${\mathcal G}=(x^\beta\eb_i, g, k_g)\in B$, such that:
\begin{enumerate}

\item $x^\beta\eb_i \mid x^{\gamma+\alpha}\eb_i$, and

\item $\n({\mathcal F})<\n({\mathcal G})$, i.e. $k_f < k_g$.

\end{enumerate}
\end{define}

The Rewritten Criterion is given as follow.

\begin{criterion} \caption{\bf --- Rewritten Criterion}
  \smallskip
  \medskip

  \noindent
Let $[{\mathcal F}, {\mathcal G}]:=(u, {\mathcal F}, v, {\mathcal
G})$ be the critical pair of ${\mathcal F}$ and ${\mathcal G}$,
where $u, v$ are monomials in $X$ such that $u\lm({\mathcal
F})=v\lm({\mathcal G})=\lcm(\lpp({\mathcal F}), \lpp({\mathcal
G}))$ and $u{\mathcal F} \sul v{\mathcal G}$. And $B\subset L[X]$
is a set of labeled polynomials. If either $u{\mathcal F}$ or
$v{\mathcal G}$ is {\bf rewritable} by $B$, then the critical pair
$[{\mathcal F}, {\mathcal G}]$ meets the Rewritten Criterion.
  \medskip
\end{criterion}

In F5 (F5B) algorithm, if a critical pair meets either Syzygy
Criterion or Rewritten Criterion, then it is not necessary to reduce
its corresponding s-polynomial.


\subsection{F5-Reduction}  \label{subsec-reduce}

The concept of signatures itself is not sufficient to ensure the
correctness of two new criteria. It is the special reduction
procedure that guarantees the critical pairs detected by criteria
are really useless. The same is true for other F5-like algorithms.

Let us start with the definition of F5-reduction.

\begin{define}[F5-reduction] \label{df-reduce}
Let ${\mathcal F}=(x^\alpha\fb_i, f)\in L[X]$ be a labeled
polynomial and $B\subset L[X]$ a set of labeled polynomials. The
labeled polynomial ${\mathcal F}$ is F5-reducible by $B$, if there
exists ${\mathcal G}=(x^\beta\fb_j, g)\in B$ such that:
\footnote{Deleting the conditions 3 and 4 does not affect the
correctness of algorithm, but leads to redundant
computations/reductions.}
\begin{enumerate}

\item $\lpp(g) \mid \lpp(f)$, denote $x^\gamma=\lpp(f)/\lpp(g)$
and $c=\lc(f)/\lc(g)$,

\item $\s({\mathcal F})\succ \s(cx^\gamma{\mathcal G})$, i.e.
$x^\alpha\eb_i \succ x^{\gamma+\beta}\eb_j$,

\item $x^\gamma{\mathcal G}$ is {\bf not} comparable by $B$, and

\item $x^\gamma{\mathcal G}$ is {\bf not} rewritable by $B$.
\end{enumerate}
If ${\mathcal F}$ is F5-reducible by $B$, let ${\mathcal
F}'={\mathcal F}-cx^\gamma{\mathcal G}$. Then this procedure:
${\mathcal F}\Longrightarrow_B {\mathcal F}'$ is called one step
F5-reduction. If ${\mathcal F}'$ is still F5-reducible by $B$, then
repeat this step until ${\mathcal F}'$ is not F5-reducible by $B$.
Suppose ${\mathcal F}^*$ is the final result that is not
F5-reducible by $B$. We say ${\mathcal F}$ F5-reduces to ${\mathcal
F}^*$ by $B$, and denote it as ${\mathcal F}\Longrightarrow_B^*
{\mathcal F}^*$.

\end{define}

The key of F5-reduction is the condition $\s({\mathcal F})\succ
\s(cx^\gamma{\mathcal G})$, i.e. $x^\alpha\eb_i \succ
x^{\gamma+\beta}\eb_j$, which makes F5-reduction much different from
other general reductions. The major function of this condition is to
preserve the signature of ${\mathcal F}$ during reductions. Thus a
direct result is that, if labeled polynomial ${\mathcal F}$
F5-reduces to ${\mathcal F}^*$ by $B$ (i.e. ${\mathcal
F}\Longrightarrow_B^* {\mathcal F}^*$), then the signatures of
${\mathcal F}$ and ${\mathcal F}^*$ are identical, i.e.
$$\s({\mathcal F})=\s({\mathcal F}^*).$$ This property plays a
crucial role in the main proofs for the correctness of F5B
algorithm. For convenience of reference, we describe this property
by the following proposition.

\begin{prop}[F5-reduction property] \label{prop-reduce}
If labeled polynomial ${\mathcal F}$ F5-reduce to ${\mathcal F}^*$
by set $B$, i.e. ${\mathcal F}\Longrightarrow_B^* {\mathcal F}^*$,
then there exist polynomials $p_1,\cdots, p_s\in \K[X]$ and labeled
polynomials ${\mathcal G}_1, \cdots, {\mathcal G}_s\subset B$, such
that:
$${\mathcal F}={\mathcal F}^*+p_1{\mathcal G}_1+\cdots+p_s{\mathcal
G}_s,$$ where leading power product $\lpp({\mathcal F})\succeq
\lpp(p_i{\mathcal G}_i)$ and signature $\s({\mathcal F})\sus
\s(p_i{\mathcal G}_i)$ for $1\le i\le s$. Moreover, signature
$\s({\mathcal F})= \s({\mathcal F}^*)$ and labeled polynomial
${\mathcal F}\bowtie {\mathcal F}^*$.
\end{prop}

The proof of this proposition is trivial by the definition of
F5-reduction.

\subsection{The F5 algorithm in Buchberger's style} \label{subsec-buchstyle}

With the definitions of Syzygy Criterion, Rewritten Criterion and
F5-reduction, we can simplify the F5 algorithm in Buchberger's style
(F5B algorithm).

\begin{algorithm}[!ht] \caption{\bf --- The F5 algorithm in Buchberger's style (F5B algorithm)} \label{algorithm1}
  \smallskip
  \Input{a polynomial $m$-tuple: $(f_1,\cdots,f_m)\subset K[X]^m$, and an admissible order $\prec$.}\\
  \Output{The \gr basis of the ideal $\langle f_1,\cdots,f_m\rangle \subset K[X]$.}
  \medskip

  \noindent
  \lbegin\\
  \SPC ${\mathcal F}_i \lla (\eb_i, f_i, i)$ for $i=1,\cdots,m$\\
  \SPC $k \lla m$ \SPC $\#$ to track the {\em number} of labled polynomials\\
  \SPC $B \lla \{{\mathcal F}_i  \mid  i=1,\cdots,m\}$\\
  \SPC $CP \lla \{\mbox{critical pair } [{\mathcal F}_i, {\mathcal F}_j] \mid  1\le i<j \le  m\}$\\
  \SPC  \lwhile $CP$ is not empty \ldo\\
  \SPC\SPC  $cp \lla $ select a critical pair from $CP$\\
  \SPC\SPC  $CP \lla CP \setminus \{cp\}$\\
  \SPC\SPC  \lif $cp$ meets {\bf neither} Syzygy Criterion {\bf nor} Rewritten
  Criterion,\\
  \SPC\SPC\SPC \lthen\\
  \SPC\SPC\SPC\SPC  ${\mathcal SP} \lla $ the s-polynomial of critical pair
  $cp$\\
  \SPC\SPC\SPC\SPC  ${\mathcal P} \lla $ the F5-reduction result of ${\mathcal SP}$ by $B$, i.e. ${\mathcal SP}\Longrightarrow_{B}^* {\mathcal P}$\\
  \SPC\SPC\SPC\SPC  $\n({\mathcal P}) \lla k+1$ \SPC $\#$ update
  the {\em number} of ${\mathcal P}$\\
  \SPC\SPC\SPC\SPC  \lif the polynomial of ${\mathcal P}$ is {\bf not} $0$, i.e. $\p({\mathcal P})\not= 0$,\\
  \SPC\SPC\SPC\SPC\SPC \lthen \\
  \SPC\SPC\SPC\SPC\SPC\SPC $CP \lla CP \cup \{\mbox{critical pair } [{\mathcal P}, {\mathcal Q}] \mid {\mathcal Q}\in B\}$\\
  \SPC\SPC\SPC\SPC \lendif\\
  \SPC\SPC\SPC\SPC  $k \lla k+1$\\
  \SPC\SPC\SPC\SPC  $B \lla B \cup \{{\mathcal P}\}$ \SPC $\#$ no matter whether $\p({\mathcal P})\not= 0$ or not\\
  \SPC\SPC \lendif\\
  \SPC \lendwhile\\
  \SPC \lreturn $\{\mbox{polynomial part of } {\mathcal Q} \mid {\mathcal Q}\in B\}$\\
  \lend
  \medskip
\end{algorithm}

According to the above algorithm, the {\em number} $\n({\mathcal
P})$ of labeled polynomial ${\mathcal P}$ is actually the order of
when ${\mathcal P}$ is being added to the set $B$. So the bigger
$\n({\mathcal P})$ is, the later ${\mathcal P}$ is generated. Notice
that the {\em numbers} of labeled polynomials in the set $B$ are
distinct from each other.

The strategy of selecting critical pairs is not specified in the F5B
algorithm, instead we simply use $$cp \lla \mbox{ select a critical
pair from } CP,$$ since the new proof proposed in next section does
not depend on the specifical strategies. Moveover, we have shown in
\citep{SunWang10} that the original F5 algorithm differs from F5B
algorithm only by a strategy of selecting critical pairs, so the
proof for the correctness of F5B algorithm can also prove the
correctness of the original F5 (or F5-like) algorithm. So next, we
focus on proving the correctness of F5B algorithm.

\section{A New Proof for the Correctness of F5B Algorithm} \label{sec-proof}

The main work of this section is to prove the correctness of F5B
algorithm presented in last section, i.e. show that the outputs of
F5B algorithm construct a \gr basis of the ideal $\langle f_1,
\cdots, f_m\rangle\subset \K[X]$.

This section is organized as follows. First, we show the difficult
point in the whole proofs by a toy example; second, we sketch the
structure of proofs and prove the main theorem; at last, we provide
the detail proofs for the lemmas and propositions used in the proof
of main theorem.

\subsection{The Thorny Problem} \label{sec-thornyproblem}

There exists a very interesting thing in F5B (or F5) algorithm. That
is, when a critical pair is detected and discarded by the two
criteria, this critical pair is usually not useless at that time
(i.e. its s-polynomial cannot F5-reduce to $0$ by the corresponding
set $B$), but when the algorithm terminates, this detected critical
pair becomes really redundant (i.e. its s-polynomial F5-reduce to
$0$ by the final set $B$). This indicates that the two criteria of
F5 algorithm can detect unnecessary computations/reductions in
advance. This is so amazing and becomes a big thorny problem in the
correctness proof of F5B algorithm.

This phenomenon happens frequently, particularly in non-homogeneous
systems. Let us see a toy example first. In order to highlight this
peculiar phenomenon, a special strategy of selecting critical pair
is used.

\begin{example}
Compute the \gr basis of the following system in $\Q[x, y, z]$
with Graded Reverse Lex Order ($x\succ y\succ z$) by F5B
algorithm:
$$\left\{\begin{array}{l} f_1=y^2+yz-x, \\
f_2=y^2-z^2+z.
\end{array}\right. $$
\end{example}

The strategy of selecting critical pairs in this toy example is:
first, find the {\bf minimal} degree of critical pairs in the set
$CP$ (the degree of critical pair $[{\mathcal F}_i, {\mathcal F}_j]$
refers to the degree of $\lcm(\lpp({\mathcal F}_i), \lpp({\mathcal
F}_i))$), and then select the {\bf maximal} critical pair from the
set $CP$ with the order $\prc$ at this minimal degree.

After initialization, the initial labeled polynomials are
$$B^{(0)}=\{{\mathcal F}_1=(\eb_1, y^2+yz-x, 1), {\mathcal
F}_2=(\eb_2, y^2-z^2+z, 2)\},$$ and critical pairs are
$$CP^{(0)}=\{[{\mathcal F}_1, {\mathcal F}_2]\}.$$

{\bf LOOP 1:} Critical pair $[{\mathcal F}_1, {\mathcal F}_2]$ is
selected from set $CP^{(0)}$. The s-polynomial of $[{\mathcal F}_1,
{\mathcal F}_2]$ is $(\eb_1, yz+z^2-x-z, 1)$ which is not
F5-reducible by set $B^{(0)}$. Then after updating the {\em number},
labeled polynomial ${\mathcal F}_3=(\eb_1, yz+z^2-x-z, 3)$ adds to
the set $B^{(0)}$. Now $$B^{(1)}=\{{\mathcal F}_1, {\mathcal F}_2,
{\mathcal F}_3\} \mbox{ and } CP^{(1)}=\{[{\mathcal F}_3, {\mathcal
F}_1], [{\mathcal F}_3, {\mathcal F}_2]\}.$$

{\bf LOOP 2:} Critical pair $[{\mathcal F}_3, {\mathcal F}_1]=(y,
{\mathcal F}_3, z, {\mathcal F}_1)$ is selected from set
$CP^{(1)}$. But labeled polynomial $z{\mathcal F}_1=(z\eb_1,
z(y^2+yz-x), 1)$ is {\bf rewritable} by set $B^{(1)}$, since there
exists labeled polynomial ${\mathcal F}_3=(\eb_1, yz+z^2-x-z, 3)$
in $B^{(1)}$ such that signature $\eb_1 \mid z\eb_1$ and {\em
number} $3 > 1$. So critical pair $[{\mathcal F}_3, {\mathcal
F}_1]$ is rejected by the {\bf Rewritten Criterion}. Now
$$B^{(2)}\{{\mathcal F}_1, {\mathcal F}_2, {\mathcal F}_3\} \mbox{ and }
CP^{(2)}=\{[{\mathcal F}_3, {\mathcal F}_2]\}.$$

{\bf LOOP 3:} Critical pair $[{\mathcal F}_3, {\mathcal F}_2]$ is
selected from set $ CP^{(2)}$. The s-polynomial of $[{\mathcal
F}_3, {\mathcal F}_2]$ is $(y\eb_1, yz^2+z^3-xy-yz-z^2, 3)$ which
F5-reduces to $(y\eb_1, -xy-yz+xz, 3)$ by set $B^{(2)}$. Then
after updating the {\em number}, labeled polynomial ${\mathcal
F}_4=(y\eb_1, -xy-yz+xz, 4)$ adds to the set $B^{(2)}$. Now
$$B^{(3)}=\{{\mathcal F}_1, {\mathcal F}_2, {\mathcal F}_3,
{\mathcal F}_4\} \mbox{ and } CP^{(3)}=\{[{\mathcal F}_4,
{\mathcal F}_1], [{\mathcal F}_4, {\mathcal F}_2], [{\mathcal
F}_4, {\mathcal F}_3]\}.$$

{\bf LOOP 4:} Critical pair $[{\mathcal F}_4, {\mathcal F}_1]=(-y,
{\mathcal F}_4, x, {\mathcal F}_1)$ is selected from set $
CP^{(3)}$. But labeled polynomial $-y{\mathcal F}_4=(y^2\eb_1,
-y(-xy-yz+xz), 4)$ is {\bf comparable} by set $B^{(3)}$, since
there exists labeled polynomial ${\mathcal F}_2=(\eb_2, y^2-z^2+z,
2)$ in $B^{(3)}$ such that leading power product $\lpp({\mathcal
F}_2)=y^2 \mid y^2$ and $\eb_1 \succ \eb_2$. So critical pair
$[{\mathcal F}_4, {\mathcal F}_1]$ is rejected by the {\bf Syzygy
Criterion}. Now
$$B^{(4)}=\{{\mathcal F}_1, {\mathcal F}_2, {\mathcal F}_3,
{\mathcal F}_4\} \mbox { and } CP^{(4)}=\{[{\mathcal F}_4,
{\mathcal F}_2], [{\mathcal F}_4, {\mathcal F}_3]\}.$$

{\bf LOOP 5:} Critical pair $[{\mathcal F}_4, {\mathcal F}_2]=(-y,
{\mathcal F}_4, x, {\mathcal F}_2)$ is selected from set $
CP^{(4)}$. But labeled polynomial $-y{\mathcal F}_4=(y^2\eb_1,
-y(-xy-yz+xz), 4)$ is {\bf comparable} by set $B^{(4)}$, since
there exists labled polynomial ${\mathcal F}_2=(\eb_2, y^2-z^2+z,
2)$ in $B^{(4)}$ such that leading power product $\lpp({\mathcal
F}_2)=y^2 \mid y^2$ and $\eb_1 \succ \eb_2$. So critical pair
$[{\mathcal F}_4, {\mathcal F}_2]$ is rejected by the {\bf Syzygy
Criterion}. Now
$$B^{(5)}=\{{\mathcal F}_1, {\mathcal F}_2, {\mathcal F}_3,
{\mathcal F}_4\} \mbox { and } CP^{(5)}=\{[{\mathcal F}_4,
{\mathcal F}_3]\}.$$

{\bf LOOP 6:} Critical pair $[{\mathcal F}_4, {\mathcal F}_3]$ is
selected from set $ CP^{(5)}$. The s-polynomial of $[{\mathcal
F}_4, {\mathcal F}_3]$ is $(yz\eb_1, -2xz^2+yz^2+x^2+xz, 4)$ which
is not F5-reducible by set $B^{(5)}$. Then after updating the {\em
number}, labeled polynomial ${\mathcal F}_5=(yz\eb_1,
-2xz^2+yz^2+x^2+xz, 5)$ adds to the set $B^{(5)}$. Now
$$B^{(6)}=\{{\mathcal F}_1, {\mathcal F}_2, {\mathcal F}_3,
{\mathcal F}_4, {\mathcal F}_5\} \mbox { and }
CP^{(6)}=\{[{\mathcal F}_5, {\mathcal F}_1], [{\mathcal F}_5,
{\mathcal F}_2], [{\mathcal F}_5, {\mathcal F}_3], [{\mathcal
F}_5, {\mathcal F}_4]\}.$$

{\bf LOOP 7:} Critical pair $[{\mathcal F}_5, {\mathcal
F}_4]=(-y/2, {\mathcal F}_5, -z^2, {\mathcal F}_4)$ is selected
from set $ CP^{(6)}$. But labeled polynomial $(-y/2){\mathcal
F}_5=(y^2z\eb_1, (-y/2)(-2xz^2+yz^2+x^2+xz), 5)$ is {\bf
comparable} by set $B^{(6)}$, since there exists labeled
polynomial ${\mathcal F}_2=(\eb_2, y^2-z^2+z, 2)$ in $B^{(6)}$
such that leading power product $\lpp({\mathcal F}_2)=y^2 \mid
y^2z$ and $\eb_1 \succ \eb_2$. So critical pair $[{\mathcal F}_5,
{\mathcal F}_4]$ is rejected by the {\bf Syzygy Criterion}. Now
$$B^{(7)}=\{{\mathcal F}_1, {\mathcal F}_2, {\mathcal F}_3,
{\mathcal F}_4, {\mathcal F}_5\} \mbox { and }
CP^{(7)}=\{[{\mathcal F}_5, {\mathcal F}_1], [{\mathcal F}_5,
{\mathcal F}_2], [{\mathcal F}_5, {\mathcal F}_3]\}.$$

{\bf LOOP 8:} Critical pair $[{\mathcal F}_5, {\mathcal
F}_3]=(-y/2, {\mathcal F}_5, xz, {\mathcal F}_3)$ is selected from
set $ CP^{(7)}$. But labeled polynomial $(-y/2){\mathcal
F}_5=(y^2z\eb_1, (-y/2)(-2xz^2+yz^2+x^2+xz), 5)$ is {\bf
comparable} by set $B^{(7)}$, since there exists labeled
polynomial ${\mathcal F}_2=(\eb_2, y^2-z^2+z, 2)$ in $B^{(7)}$
such that leading power product $\lpp({\mathcal F}_2)=y^2 \mid
y^2z$ and $\eb_1 \succ \eb_2$. So critical pair $[{\mathcal F}_5,
{\mathcal F}_3]$ is rejected by the {\bf Syzygy Criterion}. Now
$$B^{(8)}=\{{\mathcal F}_1, {\mathcal F}_2, {\mathcal F}_3,
{\mathcal F}_4, {\mathcal F}_5\} \mbox { and }
CP^{(8)}=\{[{\mathcal F}_5, {\mathcal F}_1], [{\mathcal F}_5,
{\mathcal F}_2]\}.$$

{\bf LOOP 9:} Critical pair $[{\mathcal F}_5, {\mathcal
F}_1]=(-y^2/2, {\mathcal F}_5, xz^2, {\mathcal F}_1)$ is selected
from set $ CP^{(8)}$. But labeled polynomial $(-y^2/2){\mathcal
F}_5=(y^3z\eb_1, (-y^2/2)(-2xz^2+yz^2+x^2+xz), 5)$ is {\bf
comparable} by set $B^{(8)}$, since there exists labeled
polynomial ${\mathcal F}_2=(\eb_2, y^2-z^2+z, 2)$ in $B^{(8)}$
such that leading power product $\lpp({\mathcal F}_2)=y^2 \mid
y^3z$ and $\eb_1 \succ \eb_2$. So critical pair $[{\mathcal F}_5,
{\mathcal F}_1]$ is rejected by the {\bf Syzygy Criterion}. Now
$$B^{(9)}=\{{\mathcal F}_1, {\mathcal F}_2, {\mathcal F}_3,
{\mathcal F}_4, {\mathcal F}_5\} \mbox { and }
CP^{(9)}=\{[{\mathcal F}_5, {\mathcal F}_2]\}.$$

{\bf LOOP 10:} Critical pair $[{\mathcal F}_5, {\mathcal
F}_2]=(-y^2/2, {\mathcal F}_5, xz^2, {\mathcal F}_2)$ is selected
from $ CP^{(9)}$. But labeled polynomial $(-y^2/2){\mathcal
F}_5=(y^3z\eb_1, (-y^2/2)(-2xz^2+yz^2+x^2+xz), 5)$ is {\bf
comparable} by set $B^{(9)}$, since there exists labeled
polynomial ${\mathcal F}_2=(\eb_2, y^2-z^2+z, 2)$ in $B^{(9)}$
such that leading power product $\lpp({\mathcal F}_2)=y^2 \mid
y^3z$ and $\eb_1 \succ \eb_2$. So critical pair $[{\mathcal F}_5,
{\mathcal F}_2]$ is rejected by the {\bf Syzygy Criterion}. Now
$$B^{(10)}=\{{\mathcal F}_1, {\mathcal F}_2, {\mathcal F}_3,
{\mathcal F}_4, {\mathcal F}_5\} \mbox { and }
CP^{(10)}=\emptyset.$$

Since set $CP^{(10)}$ is empty, F5B algorithm terminates and the
final set $B^{(10)}=\{{\mathcal F}_1, {\mathcal F}_2, {\mathcal
F}_3$, ${\mathcal F}_4, {\mathcal F}_5\}$. Then the polynomial set
$\{\p({\mathcal F}_1), \p({\mathcal F}_2), \p({\mathcal F}_3),
\p({\mathcal F}_4), \p({\mathcal F}_5)\}$ is a \gr basis of the
ideal generated by $\{f_1=y^2+yz-x,f_2=y^2-z^2+z\}$.

At last, we check whether the critical pairs rejected by two
criteria are really redundant. The labeled polynomial in round
bracket is the s-polynomial of corresponding critical pair. \\
{\bf LOOP 2:} $[{\mathcal F}_3, {\mathcal F}_1]=(y, {\mathcal
F}_3, z, {\mathcal F}_1)$, then
$$(y{\mathcal F}_3-z{\mathcal F}_1)-{\mathcal F}_4=(y\eb_1, 0, 3).$$
{\bf LOOP 4:} $[{\mathcal F}_4, {\mathcal F}_1]=(-y, {\mathcal
F}_4, x, {\mathcal F}_1)$, then
$$(-y{\mathcal F}_4-x{\mathcal F}_1)+2x{\mathcal
F}_3-z{\mathcal F}_1+{\mathcal F}_5=(y^2\eb_1, 0, 4).$$ {\bf LOOP
5:} $[{\mathcal F}_4, {\mathcal F}_2]=(-y, {\mathcal F}_4, x,
{\mathcal F}_2)$, then $$(-y{\mathcal F}_4-x{\mathcal
F}_2)+x{\mathcal F}_3-z{\mathcal F}_1+{\mathcal F}_5=(y^2\eb_1, 0,
4).$$ {\bf LOOP 7:} $[{\mathcal F}_5, {\mathcal F}_4]=(-y/2,
{\mathcal F}_5, -z^2, {\mathcal F}_4)$, then
$$((-y/2){\mathcal F}_5+ z^2{\mathcal F}_4)+(z^2/2){\mathcal
F}_1+(z/2){\mathcal F}_5-(x/2){\mathcal F}_4=(-y^2z\eb_1, 0, 5).$$
{\bf LOOP 8:} $[{\mathcal F}_5, {\mathcal F}_3]=(-y/2, {\mathcal
F}_5, xz, {\mathcal F}_3)$, then
$$((-y/2){\mathcal F}_5-xz{\mathcal F}_3)+(z^2/2){\mathcal
F}_1-(z/2){\mathcal F}_5-(x/2){\mathcal F}_4=(y^2z\eb_1, 0, 5).$$
{\bf LOOP 9:} $[{\mathcal F}_5, {\mathcal F}_1]=(-y^2/2, {\mathcal
F}_5, xz^2, {\mathcal F}_1)$, then
$$((-y^2/2){\mathcal F}_5-xz^2{\mathcal F}_1)+(yz^2/2-z^3/2+x^2/2+xz/2){\mathcal F}_1 + (xz^2-x^2/2){\mathcal F}_3+(z^2/2){\mathcal F}_5
=(y^3z\eb_1, 0, 5).$$ {\bf LOOP 10:} $[{\mathcal F}_5, {\mathcal
F}_2]=(-y^2/2, {\mathcal F}_5, xz^2, {\mathcal F}_2)$, then
$$((-y^2/2){\mathcal F}_5-xz^2{\mathcal
F}_2)+(yz^2/2-z^3/2+x^2/2+xz/2){\mathcal F}_1+(z^2/2){\mathcal
F}_5-(x^2/2){\mathcal F}_3=(y^3z\eb_1, 0, 5).$$ All these
s-polynomials F5-reduces to $0$ by $B^{(10)}$, so both the
criteria are correct.

\begin{remark}
Notice that the s-polynomial of $[{\mathcal F}_3, {\mathcal F}_1]$
F5-reduces to $0$ by the labeled polynomial ${\mathcal F}_4$, which
is generated in LOOP 3. However, the critical pair $[{\mathcal F}_3,
{\mathcal F}_1]$ is rejected in LOOP 2, which implies that when this
critical pair is being discarded, its s-polynomial $y{\mathcal
F}_3-z{\mathcal F}_1$ cannot F5-reduce to 0 by the set
$B^{(1)}=\{{\mathcal F}_1, {\mathcal F}_2, {\mathcal F}_3\}$.
Similar cases also happen to critical pairs $[{\mathcal F}_4,
{\mathcal F}_1]$ and $[{\mathcal F}_4, {\mathcal F}_2]$. These facts
illustrate the thorny problem mentioned earlier.
\end{remark}

This thorny problem is a big handicap for the correctness proof of
F5B (or F5) algorithm, and as we know, it is not well handled in
other existing proofs for F5 algorithm.

The new proof presented in this paper averts this thorny problem
subtly. Instead of proving the critical pairs are useless when they
are being detected, we concentrate on showing that these critical
pairs are redundant after the algorithm terminates. This is detailed
in next subsection.

\subsection{Main Theorem}

In order to show the detected critical pairs are redundant after
the algorithm terminates, we need to save these critical pairs and
discuss them afterwards. Thus, we modify F5B algorithm slightly.

\begin{algorithm}[!ht] \caption{\bf --- The F5B algorithm modified by a subtle trick (F5M algorithm)} \label{alg-f5m}
  \smallskip
  \Input{a polynomial $m$-tuple: $(f_1,\cdots,f_m)\subset K[X]^m$, and an admissible order $\prec$.}\\
  \Output{The \gr basis of the ideal $\langle f_1,\cdots,f_m\rangle \subset K[X]$.}
  \medskip

  \noindent
  \lbegin\\
  \SPC ${\mathcal F}_i \lla (\eb_i, f_i, i)$ for $i=1,\cdots,m$\\
  \SPC $k \lla m$ \SPC $\#$ to track the {\em number} of labled polynomials\\
  \SPC $B \lla \{{\mathcal F}_i  \mid  i=1,\cdots,m\}$\\
  \SPC $D \lla \emptyset$\\
  \SPC $CP \lla \{\mbox{critical pair } [{\mathcal F}_i, {\mathcal F}_j] \mid  1\le i<j \le  m\}$\\
  \SPC  \lwhile $CP$ is not empty \ldo\\
  \SPC\SPC  $cp \lla $ select a critical pair from $CP$\\
  \SPC\SPC  $CP \lla CP \setminus \{cp\}$\\
  \SPC\SPC  \lif $cp$ meets {\bf neither} Syzygy Criterion {\bf nor} Rewritten
  Criterion,\\
  \SPC\SPC\SPC \lthen\\
  \SPC\SPC\SPC\SPC  ${\mathcal SP} \lla $ the s-polynomial of critical pair
  $cp$\\
  \SPC\SPC\SPC\SPC  ${\mathcal P} \lla $ the F5-reduction result of ${\mathcal SP}$ by $B$,
  i.e. ${\mathcal SP}\Longrightarrow_{B}^* {\mathcal P}$\\
  \SPC\SPC\SPC\SPC  $\n({\mathcal P}) \lla k+1$ \SPC $\#$ update
  the {\em number} of ${\mathcal P}$\\
  \SPC\SPC\SPC\SPC  \lif the polynomial of ${\mathcal P}$ is {\bf not} $0$, i.e. $\p({\mathcal P})\not= 0$,\\
  \SPC\SPC\SPC\SPC\SPC \lthen \\
  \SPC\SPC\SPC\SPC\SPC\SPC $CP \lla CP \cup \{\mbox{critical pair } [{\mathcal P}, {\mathcal Q}] \mid {\mathcal Q}\in B\}$\\
  \SPC\SPC\SPC\SPC \lendif\\
  \SPC\SPC\SPC\SPC  $k \lla k+1$\\
  \SPC\SPC\SPC\SPC  $B \lla B \cup \{{\mathcal P}\}$ \SPC $\#$ no matter whether $\p({\mathcal P})\not= 0$ or not\\
  \SPC\SPC\SPC \lelse\\
  \SPC\SPC\SPC\SPC  $D \lla D \cup \{cp\}$ \SPC $\#$ save the detected critical pairs\\
  \SPC\SPC \lendif\\
  \SPC \lendwhile\\
  \SPC \lreturn $\{\mbox{polynomial part of } {\mathcal Q} \mid {\mathcal Q}\in B\}$\\
  \lend
  \medskip
\end{algorithm}

The only difference between the F5B algorithm and F5M algorithm is:
the detected critical pairs are all saved in set $D$. For
convenience, we use the notations $B_{end}$ and $D_{end}$ to express
the corresponding sets $B$ and $D$ when the F5M algorithm
terminates.

Since initial polynomial set $\{f_1, \cdots, f_m\}= \{\p({\mathcal
Q})\mid {\mathcal Q}\in B_0\}$ and $B_0 \subset B_{end}$ by the
F5M algorithm, our main purpose of this paper is to prove the
following correctness theorem.

\begin{theorem}[Correctness Theorem] \label{thm-correctness}
The set $\{\p({\mathcal Q}) \mid {\mathcal Q}\in B_{end}\} \subset
\K[X]$ itself is a \gr basis.
\end{theorem}

To prove this theorem, we need a powerful tool: $t$-representation
for labeled polynomials.

\begin{define}[$t$-representation] \label{df-trep}
Let ${\mathcal F}\in L[X]$ be a labeled polynomial, $B\subset L[X]$
a set of labeled polynomials and $t\in PP(X)$ a power product. We
say labeled polynomial ${\mathcal F}$ has a $t$-representation
w.r.t. set $B$, if there exist polynomials $p_1,\cdots,p_s \in K[X]$
and labeled polynomials ${\mathcal G}_1,\cdots,{\mathcal G}_s\in B$,
such that:
$$\p({\mathcal F})=p_1\p({\mathcal G}_1)+\cdots+p_s\p({\mathcal
G}_s),$$ where labeled polynomial ${\mathcal F}\unrhd p_i{\mathcal
G}_i$ and power product $t\succeq \lpp(p_i{\mathcal G}_i)$ for
$i=1,\cdots,s$.
\end{define}

Compared with the definition of $t$-representation in polynomial
version, the $t$-representation for labeled polynomials has an extra
condition ${\mathcal F}\unrhd p_i{\mathcal G}_i$ on the {\em
signatures} and {\em numbers}.

For convenience, we say {\em the critical pair $[{\mathcal F},
{\mathcal G}]=(u, {\mathcal F}, v, {\mathcal G})$ has a
$t$-representation w.r.t. set $B$}, if the s-polynomial of
$[{\mathcal F}, {\mathcal G}]$ has a $t$-representation w.r.t. set
$B$ where $t\prec \lcm(\lpp({\mathcal F}), \lpp({\mathcal G}))$.

The following theorem is the main result on $t$-representation for
labeled polynomials. Its proof is straight from its polynomial
version, so we omit the detail proof here. For interesting readers,
please see \citep{Becker93}.

\begin{theorem}[$t$-representation] \label{thm-gbt}
Let $B\subset L[X]$ be a set of labeled polynomials. If for all
labeled polynomials ${\mathcal F}, {\mathcal G}\in B$, critical
pair $[{\mathcal F}, {\mathcal G}]$ always has a
$t$-representation w.r.t. set $B$, then the polynomial set
$\{\p({\mathcal P})\mid {\mathcal P}\in B\}\subset K[X]$ itself is
a \gr basis.
\end{theorem}

So far, in order to prove the Correctness Theorem
\ref{thm-correctness}, it suffices to show that for any labeled
polynomials ${\mathcal F}, {\mathcal G} \in B_{end}$, the critical
pair $[{\mathcal F}, {\mathcal G}]$ always has a $t$-representation
w.r.t. set $B_{end}$. In fact, if we examine all these critical
pairs in detail, there are only two kinds of critical pairs
generated by set $B_{end}$:
\begin{enumerate}
\item The ones that have been operated during the loops, i.e.
their s-polynomials have been calculated and then F5-reduced. These
F5-reduction results have added to set $B_{end}$.

\item The ones detected by either Syzygy Criterion or Rewritten
Criterion. In F5M algorithm, all these critical pairs have been
collected into set $D_{end}$.
\end{enumerate}

For the first kind of critical pairs, the following proposition,
which is proved in next subsection, ensures that these critical
pairs have $t$-representations w.r.t. set $B_{end}$.

\begin{prop}[first kind] \label{prop-firstkind}
If a critical pair is operated during the loops, i.e. it is not
detected by the two criteria, then it has a $t$-representation
w.r.t. set $B_{end}$.
\end{prop}

For the second kind of critical pairs, the proof that they have
$t$-representations w.r.t. set $B_{end}$ is a bit complicated. In
fact, we cannot show this directly, since an extra condition is
necessary.

Let ${\mathcal F}, {\mathcal G}\in L[X]$ be two labeled
polynomials and $B\subset L[X]$ a set of labeled polynomials. We
say {\em all the lower critical pairs of $[{\mathcal F}, {\mathcal
G}]$ have $t$-representations w.r.t. set $B$}, if for any critical
pair $[{\mathcal P}, {\mathcal Q}]$ such that $[{\mathcal P},
{\mathcal Q}] \prl [{\mathcal F}, {\mathcal G}]$ where ${\mathcal
P}, {\mathcal Q}\in B$, the critical pair $[{\mathcal P},
{\mathcal Q}]$ always has a $t$-representation w.r.t. set $B$.

The following theorem shows the second kind of critical pairs have
$t$-representations w.r.t. set $B_{end}$ with an extra condition.

\begin{theorem}[second kind] \label{thm-secondkind}
Let $[{\mathcal F}, {\mathcal G}]=(u, {\mathcal F}, v, {\mathcal
G})$ be a critical pair, where ${\mathcal F}, {\mathcal G} \in
B_{end}$ and $u, v$ are monomials in $X$ such that $u\lm({\mathcal
F})=v\lm({\mathcal G})=\lcm(\lpp({\mathcal F}), \lpp({\mathcal
G}))$. Then the critical pair $[{\mathcal F}, {\mathcal G}]$ has a
$t$-representation w.r.t. set $B_{end}$, if
\begin{enumerate}
\item labeled polynomial $u{\mathcal F}$ (or $v{\mathcal G}$) is
either comparable or rewritable by $B_{end}$, and

\item all the lower critical pairs of $[{\mathcal F}, {\mathcal
G}]$ have $t$-representations w.r.t. set $B_{end}$.
\end{enumerate}
\end{theorem}

With Proposition \ref{prop-firstkind} (first kind) and Theorem
\ref{thm-secondkind} (second kind), we are now able to prove the
Correctness Theorem \ref{thm-correctness}. The extra condition in
Theorem \ref{thm-secondkind} is satisfied subtly.

\begin{thm31}
The set $\{\p({\mathcal Q}) \mid {\mathcal Q}\in B_{end}\} \subset
\K[X]$ itself is a \gr basis.
\end{thm31}

\begin{proof}
Let $CP_{all}$ be the set of all critical pairs generated by set
$B_{end}$. Then all the critical pairs in $CP_{all}\setminus
D_{end}$ have $t$-representations w.r.t. $B_{end}$ by Proposition
\ref{prop-firstkind} (first kind). Next, it only remains to show
that critical pair $cp$ has a $t$-representation w.r.t. set
$B_{end}$ for all $cp\in D_{end}$.

The strategy of the proof is as follows.
\begin{enumerate}

\item[(1)] Select the minimal critical pair, say $cp_{min}$, from
set $D_{end}$ w.r.t. the order $\sul$.

\item[(2)] Show the critical pair $cp_{min}$ has a $t$-representation
w.r.t. set $B_{end}$.

\item[(3)] Remove the critical pair $cp_{min}$ from set $D_{end}$.
\end{enumerate}
If set $D_{end}$ is not empty, then repeat the steps (1), (2) and
(3). Since the cardinality of set $D_{end}$ is finite, this
procedure terminates after finite steps. If all the critical pairs
in set $D_{end}$ are proved in this way, the theorem is proved.

The steps (1) and (3) are trivial, so it only needs to show how the
step (2) is done. Since critical pair $cp_{min}$ is the minimal one
in set $D_{end}$, then all the critical pairs which are lower than
$cp_{min}$ should be contained in the set $CP_{all}\setminus
D_{end}$ and hence have $t$-representations w.r.t. set $B_{end}$
(because set $D_{end}$ contains all the unproved critical pairs).
Critical pair $cp_{min}\in D_{end}$ also means $cp_{min}$ meets
either Syzygy Criterion or Rewritten Criterion, so the critical pair
$cp_{min}$ has a $t$-representation w.r.t. set $B_{end}$ by Theorem
\ref{thm-secondkind} (second kind).

After all, the critical pairs in $CP_{all}$ all have
$t$-representations w.r.t. set $B_{end}$. Then the polynomial set
$\{\p({\mathcal P}) \mid  {\mathcal P}\in B_{end}\}$ itself is a \gr
basis by Theorem \ref{thm-gbt} ($t$-representation).
\end{proof}

The proof of Proposition \ref{prop-firstkind} (first kind) for the
first kind of critical pairs is simple. However, the proof of
Theorem \ref{thm-secondkind} (second kind) for the second kind of
critical pairs is quite complicated. Next, we sketch the idea of
this proof. All the following lemmas and propositions are proved in
next subsection. We begin by an important definition.

\begin{define}[strictly lower representation]
Let ${\mathcal F}\in L[X]$ be a labeled polynomial and $B\subset
L[X]$ a set of labeled polynomials. We say labeled polynomial
${\mathcal F}$ has a strictly lower representation w.r.t. set $B$,
if there exist polynomials $p_1,\cdots,p_s \in K[X]$ and labeled
polynomials ${\mathcal G}_1, \cdots, {\mathcal G}_s\in B$, such
that: $$ \p({\mathcal F})=p_1\p({\mathcal
G}_1)+\cdots+p_s\p({\mathcal G}_s),$$ where labeled polynomial
${\mathcal F} \sul p_i{\mathcal G}_i$ for $i=1,\cdots,s$.
\end{define}

Compared with the $t$-representation defined earlier, the strictly
lower representation does not need the constraints on the leading
power products $\lpp(p_i{\mathcal G}_i)$. Besides, the relation
``$\unrhd$" in Definition \ref{df-trep} ($t$-representation) becomes
``$\sul$" here, which is why we name it as {\em strictly lower
representation}.

By the above definition, we first have two propositions on {\em
comparable} and {\em rewritable}.

\begin{prop}[comparable] \label{prop-syz}
Let ${\mathcal F}\in B_{end}$ be a labeled polynomial and
$cx^\gamma$ a non-zero monomial in $X$. If labeled polynomial
$cx^\gamma{\mathcal F}$ is {\bf comparable} by $B_{end}$, then
$cx^\gamma{\mathcal F}$ has a strictly lower representation w.r.t.
set $B_{end}$.
\end{prop}

\begin{prop}[rewritable] \label{prop-rew}
Let ${\mathcal F}\in B_{end}$ be a labeled polynomial and
$cx^\gamma$ a non-zero monomial in $X$. If labeled polynomial
$cx^\gamma{\mathcal F}$ is {\bf rewritable} by $B_{end}$, then
$cx^\gamma{\mathcal F}$ has a strictly lower representation w.r.t.
set $B_{end}$.
\end{prop}

Next, the key lemma connect the {\em strictly lower
representation} and {\em $t$-representation}. We say {\em all the
lower critical pairs of ${\mathcal F}$ have $t$-representations
w.r.t. set $B$}, where ${\mathcal F}$ is a labeled polynomial and
$B$ is a set of labeled polynomials, if for all critical pairs
$[{\mathcal P}, {\mathcal Q}]=(r, {\mathcal P}, t, {\mathcal Q})$
such that $r{\mathcal P}\prl {\mathcal F}$ where ${\mathcal
P},{\mathcal Q}\in B$, the critical pair $[{\mathcal P}, {\mathcal
Q}]$ always has a $t$-representation w.r.t. set $B$.

\begin{lemma}[key lemma] \label{lem-key}
Let ${\mathcal F}\in L[X]$ be a labeled polynomial. If
\begin{enumerate}

\item labeled polynomial ${\mathcal F}$ has a strictly lower
representation w.r.t. set $B_{end}$, and

\item all the lower critical pairs of ${\mathcal F}$ have
$t$-representations w.r.t. set $B_{end}$.
\end{enumerate}
Then labeled polynomial ${\mathcal F}$ has a $t$-representation
w.r.t. set $B_{end}$ where $t=\lpp({\mathcal F})$. Furthermore,
there exists a labeled polynomial ${\mathcal H}\in B_{end}$ such
that: $\lpp({\mathcal H})\mid \lpp({\mathcal F})$ and ${\mathcal
F}\sul x^\lambda{\mathcal H}$ where $x^\lambda=\lpp({\mathcal
F})/\lpp({\mathcal H})$.
\end{lemma}

Based on Lemma \ref{lem-key} (key lemma), it is esay to obtain the
following two propositions. Please pay attention to the position of
the labeled polynomial ${\mathcal F}$ in the critical pair of each
proposition.

\begin{prop}[left] \label{prop-left}
Let $[{\mathcal F}, {\mathcal G}]=(u, {\mathcal F}, v, {\mathcal
G})$ be a critical pair, where ${\mathcal F}, {\mathcal G}\in
B_{end}$ are labeled polynomials and $u, v$ are monomials in $X$
such that $u\lm({\mathcal F})=v\lm({\mathcal
G})=\lcm(\lpp({\mathcal F})$, $\lpp({\mathcal G}))$. Then the
critical pair $[{\mathcal F}, {\mathcal G}]$ has a
$t$-representation w.r.t. set $B_{end}$, if
\begin{enumerate}
\item labeled polynomial $u{\mathcal F}$ has a strictly lower
representation w.r.t. set $B_{end}$, and

\item all the lower critical pairs of $[{\mathcal F}, {\mathcal
G}]$ have $t$-representations w.r.t. set $B_{end}$.
\end{enumerate}
\end{prop}

\begin{prop}[right] \label{prop-right}
Let $[{\mathcal G}, {\mathcal F}]=(v, {\mathcal G}, u, {\mathcal
F})$ be a critical pair, where ${\mathcal G}, {\mathcal F}\in
B_{end}$ are labeled polynomials and $v, u$ are monomials in $X$
such that $v\lm({\mathcal G})=u\lm({\mathcal F})=\lcm(\lpp({\mathcal
G})$, $\lpp({\mathcal F}))$. Then the critical pair $[{\mathcal G},
{\mathcal F}]$ has a $t$-representation w.r.t. set $B_{end}$, if
\begin{enumerate}
\item labeled polynomial $u{\mathcal F}$ has a strictly lower
representation w.r.t. set $B_{end}$, and

\item all the lower critical pairs of $[{\mathcal G}, {\mathcal
F}]$ have $t$-representations w.r.t. set $B_{end}$.
\end{enumerate}
\end{prop}

Now, combined with Propositions \ref{prop-syz} (comparable),
\ref{prop-rew} (rewritable), \ref{prop-left} (left) and
\ref{prop-right} (right), Theorem \ref{thm-secondkind} (second kind)
is proved.

\subsection{Proofs of Lemmas and Propositions} \label{subsec-lempropproof}

In this subsection, we list the detail proofs for the lemmas and
propositions appearing in last subsection.

\begin{prop36}
If a critical pair is operated during the loops, i.e. it is not
detected by the two criteria, then it has a $t$-representation
w.r.t. set $B_{end}$.
\end{prop36}

\begin{proof}
Let $cp=[{\mathcal F}, {\mathcal G}]$ be a critical pair which is
not rejected by the two criteria. Assume $cp$ is being selected in
the $l$th loop (from set $CP^{(l-1)}$) and $B^{(l-1)}$ is the
labeled polynomial set before the $l$th loop begins.

Since critical pair $cp$ is not rejected by two criteria, its
s-polynomial is calculated and F5-reduces by set ${B^{(l-1)}}$ to a
new labeled polynomial ${\mathcal P}$, i.e. $\sp({\mathcal F},
{\mathcal G})\Longrightarrow_{B^{(l-1)}}^* {\mathcal P}$. Next, only
two possibilities may happen to the labeled polynomial ${\mathcal
P}$.
\begin{enumerate}

\item If $\p({\mathcal P})=0$, it is easy to check that the
s-polynomial $\sp({\mathcal F}, {\mathcal G})$ of $[{\mathcal F},
{\mathcal G}]$ has a $t$-representation w.r.t. set $B^{(l-1)}$ where
$t=\lpp(\sp({\mathcal F}, {\mathcal G}))$ by the definition of
F5-reduction and hence $t\prec \lcm(\lpp({\mathcal F}),
\lpp({\mathcal G}))$.

\item If $\p({\mathcal P})\not=0$, then the {\em number} of
${\mathcal P}$ is updated and denote this new labeled polynomial as
${\mathcal P}'$. Since signature $\s(\sp({\mathcal F},{\mathcal
G}))= \s({\mathcal P})= \s({\mathcal P}')$ and the {\em number}
$\n(\sp({\mathcal F},{\mathcal G}))=\n({\mathcal P})< \n({\mathcal
P}')$, then labeled polynomial $\sp({\mathcal F},{\mathcal G})\sul
{\mathcal P}'$ by the definition of ``$\sul$". Therefore, the
s-polynomial $\sp({\mathcal F},{\mathcal G})$ has a
$t$-representation w.r.t. set $B^{(l-1)} \cup \{{\mathcal P}'\}$
where $t=\lpp(\sp({\mathcal F}, {\mathcal G}))\prec
\lcm(\lpp({\mathcal F}),  \lpp({\mathcal G}))$. Notice that set
$B^{(l)}=B^{(l-1)} \cup \{{\mathcal P}'\}$ by the algorithm and both
$B^{(l-1)}, B^{(l)}\subset B_{end}$.
\end{enumerate}
Thus in either of the above cases, the critical pair $[{\mathcal
F},{\mathcal G}]$ has a $t$-representation w.r.t. set $B_{end}$.
\end{proof}

Next, we begin the proofs for Theorem \ref{thm-secondkind} (second
kind). The following lemma reveals the meanings of signatures and it
is also used in the proof of Proposition \ref{prop-syz} (comparable)
and \ref{prop-rew} (rewritable).

\begin{lemma}[signature] \label{lem-signature}
If labeled polynomial ${\mathcal F}=(x^\alpha\eb_j, f, k)\in
B_{end}$, then
$$f=cx^\alpha f_j+p_1\p({\mathcal G}_1)+\cdots+p_s\p({\mathcal G}_s),$$
where $c$ is a non-zero constant in  $\K$, $p_i\in \K[X]$ and
${\mathcal G}_i\in B_{end}$ such that either $p_i=0$ or signature
$\s({\mathcal F}) \sus \s(p_i{\mathcal G}_i)$ for $i=1,\cdots,s$.
\end{lemma}

\begin{proof}
We prove this proposition by induction of the loop $l$. Let
$B^{(l-1)}$ be the labeled polynomial set before the $l$th loop
begins and $B^{(l)}$ the labeled polynomial set when the $l$th
loop is over.

First, when $l=0$, consider the set $B^{(0)}=\{(\eb_i, f_i, i)
 \mid  i=1,\cdots,m\} $ where $f_i$'s are initial polynomials.
Clearly, $$f_i=f_i,$$ which shows the proposition holds for the
set $B^{(0)}$.

Second, suppose the proposition holds for the set $B^{(l-1)}$. Then
the next goal is to show the proposition holds for the set
$B^{(l)}$. Denote the critical pair that is selected (from set
$CP^{(l-1)}$) in the $l$th loop as $cp=[{\mathcal Q}_1, {\mathcal
Q}_2]=(u_1, {\mathcal Q}_1, u_2, {\mathcal Q}_2)$, where ${\mathcal
Q}_1, {\mathcal Q}_2\in B^{(l-1)}$ and $u_1, u_2$ are monomials in
$X$ such that $u_1\lm({\mathcal Q}_1)=u_2\lm({\mathcal
Q}_2)=\lcm(\lpp({\mathcal Q}_1), \lpp({\mathcal Q}_2))$.

If critical pair $cp$ meets either of criteria, then this critical
pair is discarded and no labeled polynomial adds to set $B^{(l-1)}$,
which means $B^{(l)}=B^{(l-1)}$. Then the proposition holds for set
$B^{(l)}$.

It remains to show that when the critical pair $cp$ does not meet
either of criteria, the proposition still holds for set $B^{(l)}$.
In this case, the s-polynomial $\sp({\mathcal Q}_1, {\mathcal Q}_2)$
is calculated and F5-reduces to a new labeled polynomial ${\mathcal
P}$ by the set $B^{(l-1)}$, i.e. $\sp({\mathcal Q}_1, {\mathcal
Q}_2)\Longrightarrow_{B^{(l-1)}}^* {\mathcal P}$. Then the {\em
number} of ${\mathcal P}$ is updated and denote this new labeled
polynomial as ${\mathcal P}'$. Clearly, signature $\s({\mathcal
P})=\s({\mathcal P}')$ and polynomial $\p({\mathcal P})=\p({\mathcal
P}')$. Next, $B^{(l)}= B^{(l-1)} \cup \{{\mathcal P}'\}$ by the
algorithm. Therefore, it suffices to prove that the proposition
holds for ${\mathcal P}'$.

By Proposition \ref{prop-reduce} (F5-reduction property), as
s-polynomial $\sp({\mathcal Q}_1, {\mathcal
Q}_2)\Longrightarrow_{B^{(l-1)}}^* {\mathcal P}$, there exist
polynomials $p_1,\cdots,p_s \in \K[X]$ and labled polynomials
${\mathcal G}_1, \cdots, {\mathcal G}_s\in B^{(l-1)}$, such that
${\mathcal P}=\sp({\mathcal Q}_1, {\mathcal Q}_2) +p_1{\mathcal
G}_1+\cdots+p_s{\mathcal G}_s$, where signature $\s({\mathcal
P})=\s(\sp({\mathcal Q}_1, {\mathcal Q}_2))\sus \\ \s(p_i{\mathcal
G}_i)$ for $i=1,\cdots,s$. Notice that s-polynomial $\sp({\mathcal
Q}_1, {\mathcal Q}_2)=u_1{\mathcal Q}_1-u_2{\mathcal Q}_2$. The
above equation equals to
\begin{equation}\label{eq_6}
{\mathcal P}=u_1{\mathcal Q}_1-u_2{\mathcal Q}_2+p_1{\mathcal
G}_1+\cdots+p_s{\mathcal G}_s.\end{equation} The definition of
critical pair $[{\mathcal Q}_1, {\mathcal Q}_2]$ shows
$u_1{\mathcal Q}_1\sul u_2{\mathcal Q}_2$. As labeled polynomial
$u_1{\mathcal Q}_1$ is not rewritable by $B^{(l-1)}$, then
signature $\s(u_1{\mathcal Q}_1)\sus \s(u_2{\mathcal Q}_2)$ holds;
otherwise $u_1{\mathcal Q}_1$ is rewritable by $\{{\mathcal
Q}_2\}\subset B^{(l-1)}$. Therefore, according to the addition of
labeled polynomials, signature $\s({\mathcal P}')=\s({\mathcal
P})= \s(u_1{\mathcal Q}_1)= \s(\sp({\mathcal Q}_1, {\mathcal
Q}_2))\sus \s(p_i{\mathcal G}_i)$ for $i=1,\cdots,s$ and
$\s({\mathcal P}')=\s({\mathcal P})=\s(u_1{\mathcal Q}_1)\sus
\s(u_2{\mathcal Q}_2)$.

Now consider the polynomial part of equation (\ref{eq_6}):
\begin{equation}\label{eq_11}
\p({\mathcal P}')=\p({\mathcal P})=u_1\p({\mathcal
Q}_1)-u_2\p({\mathcal Q}_2)+p_1\p({\mathcal
G}_1)+\cdots+p_s\p({\mathcal G}_s).
\end{equation}
Since labeled polynomial ${\mathcal Q}_1\in B^{(l-1)}$, assume
${\mathcal Q}_1=(x^\gamma \eb_j, q, k')$, by the induction
hypothesis,
$$\p({\mathcal Q}_1)=cx^\gamma f_j+q_1\p({\mathcal H}_1)+\cdots+q_r\p({\mathcal H}_r),$$
where $c$ is a non-zero constant in $\K$, $q_i\in \K[X]$ and
${\mathcal H}_i\in B^{(l-1)}$ such that either $q_i=0$ or signature
$\s({\mathcal Q}_1)\sus \s(q_i{\mathcal H}_i)$ for $i=1,\cdots,r$.
Since $u_1$ is a non-zero monomial in $X$, signature $\s({\mathcal
P}')= \s(u_1{\mathcal Q}_1)=\lpp(u_1)x^\gamma \eb_j$. Substitute the
above expression of $\p({\mathcal Q}_1)$ back into equation
(\ref{eq_11}), then a new representation of $\p({\mathcal P}')$ is
obtained, which shows that the proposition holds for set $B^{(l)}$.
Then the proposition is proved.
\end{proof}


The above lemma explains the implications of the signatures, i.e.
for any labeled polynomial ${\mathcal F}=(x^\alpha\eb_j, f, k)\in
B_{end}$, its polynomial $f$ is F5-reduced from the polynomial
$x^\alpha f_j$, where $f_j$ is an initial polynomial. In fact, this
lemma holds more generally.

\begin{cor}[signature]\label{cor-signature}
Let ${\mathcal F}=(x^\alpha\eb_j, f, k)\in B_{end}$ be a labeled
polynomial and $cx^\gamma$ a non-zero monomial in $X$. For the
labeled polynomial $cx^\gamma{\mathcal F}=(x^{\gamma+\alpha}\eb_j,
cx^\gamma f, k)$, then
$$cx^\gamma f=\bar{c}x^{\gamma+\alpha} f_j+p_1\p({\mathcal G}_1)+\cdots+p_s\p({\mathcal G}_s),$$
where $\bar{c}$ is a non-zero constant in  $\K$, $p_i\in \K[X]$
and ${\mathcal G}_i\in B_{end}$ such that either $p_i=0$ or
signature $\s(cx^\gamma{\mathcal F}) \sus \s(p_i{\mathcal G}_i)$
for $i=1,\cdots,s$.
\end{cor}

With a little care, the representations in Lemma \ref{lem-signature}
(signature) and Corollary \ref{cor-signature} (signature) only
constrain the signatures of ${\mathcal F}$ and $p_i{\mathcal G}_i$,
and do not limit the leading power products $\lpp({\mathcal F})$ and
$\lpp(p_i{\mathcal G}_i)$.

Remark that Lemma \ref{lem-signature} (signature) itself is not
sufficient to provide a strictly lower representation for the
labeled polynomial ${\mathcal F}$, since signature $\s({\mathcal
F})=x^\alpha \eb_j= \s(x^\alpha {\mathcal F}_j)$ but the {\em
number} $\n({\mathcal F})\ge \n(x^\alpha {\mathcal F}_j)$, which
means labeled polynomial ${\mathcal F} \unlhd x^\alpha {\mathcal
F}_j$, where ${\mathcal F}_j$ is the labeled polynomial of initial
polynomial $f_j$.

The following two propositions show that if a labeled polynomial is
either {\em comparable} or {\em rewritable} by $B_{end}$, then this
labeled polynomial has a {\em strictly lower representation} w.r.t.
set $B_{end}$.

\begin{prop38}
Let ${\mathcal F}=(x^\alpha\eb_j, f, k_f)\in B_{end}$ be a labeled
polynomial and $cx^\gamma$ a non-zero monomial in $X$. If labeled
polynomial $cx^\gamma{\mathcal F}$ is {\bf comparable} by
$B_{end}$, then $cx^\gamma{\mathcal F}$ has a strictly lower
representation w.r.t. set $B_{end}$.
\end{prop38}

\begin{proof}
Since $cx^\gamma {\mathcal F}=(x^{\gamma+\alpha}\eb_j, cx^\gamma
f, k_f)$ is comparable by $B_{end}$, there exists labeled
polynomial ${\mathcal G}=(x^\beta \eb_l, g, k_g)\in B_{end}$ such
that (1) $\lpp(g) \mid x^{\gamma+\alpha}$ and (2) $\eb_j \succ
\eb_l$. Denote $x^\lambda=x^{\gamma+\alpha}/\lpp(g)$, then
$x^{\gamma+\alpha}=x^\lambda \lpp(g)$. Let ${\mathcal F}_j=(\eb_j,
f_j, j)\in B_{end}$ be the labeled polynomial of initial
polynomial $f_j$. Then the polynomial 2-tuple $(g, -f_j)$ is a
principle syzygy of the 2-tuple vector $(f_j, g)$ in free module
$(\K[X])^2$. That is
$$gf_j-f_jg = 0 \mbox{ and } \lm(g)f_j=f_jg-(g-\lm(g))f_j.$$
As $x^{\gamma+\alpha}=x^\lambda \lpp(g)$, then
\begin{equation}\label{eq_17}
x^{\gamma+\alpha} f_j=x^\lambda \lpp(g)
f_j=\frac{x^\lambda}{\lc(g)} (f_jg-(g-\lm(g))f_j) =
\frac{x^\lambda}{\lc(g)}
f_jg-\frac{x^\lambda}{\lc(g)}(g-\lm(g))f_j$$ $$ = q_1 g+ q_2 f_j=
q_1\p({\mathcal G})+q_2\p({\mathcal F}_j),
\end{equation}
where $q_1=\frac{x^\lambda}{\lc(g)} f_j$ and
$q_2=-\frac{x^\lambda}{\lc(g)} (g-\lm(g))$.

As $\eb_j \succ \eb_l$ holds by hypothesis, then labeled polynomial
$cx^\gamma {\mathcal F} \sul q_1 {\mathcal G}$. Also labeled
polynomial $cx^\gamma {\mathcal F}\sul q_2{\mathcal F}_j$, as the
signature $\s(cx^\gamma {\mathcal F})= x^{\gamma+\alpha}\eb_j =
x^\lambda \lpp(g)\eb_j \succ x^\lambda \lpp(g-\lm(g))\eb_j =
\lpp(q_2)\eb_j = \s(q_2{\mathcal F}_j)$.

Since labeled polynomial ${\mathcal F}\in B_{end}$ and $cx^\gamma$
is a non-zero monomial, Corollary \ref{cor-signature} (signature)
shows
\begin{equation}\label{eq_18}
\p(cx^\gamma {\mathcal F})=cx^\gamma
f=\bar{c}x^{\gamma+\alpha} f_j+p_1\p({\mathcal
H}_1)+\cdots+p_s\p({\mathcal H}_s),
\end{equation}
where $\bar{c}$ is a non-zero constant in  $\K$, $p_i\in \K[X]$
and ${\mathcal H}_i\in B_{end}$ such that either $p_i=0$ or
signature $\s(cx^\gamma{\mathcal F}) \sus \s(p_i{\mathcal H}_i)$
and hence labeled polynomial $cx^\gamma {\mathcal F} \sul
p_i{\mathcal H}_i$ for $i=1,\cdots,s$.

Substitute the expression of polynomial $x^{\gamma+\alpha} f_j$ in
equation (\ref{eq_17}) into (\ref{eq_18}). Then
$$\p(cx^\gamma {\mathcal F})=\bar{c} q_1\p({\mathcal
G})+\bar{c} q_2\p({\mathcal F}_j)+p_1\p({\mathcal
H}_1)+\cdots+p_s\p({\mathcal H}_s),$$ where labeled polynomial
$cx^\gamma {\mathcal F} \sul \bar{c}q_1{\mathcal G}$, $cx^\gamma
{\mathcal F} \sul \bar{c}q_2{\mathcal F}_j$ and $cx^\gamma {\mathcal
F} \sul p_i{\mathcal H}_i$ for $i=1,\cdots,s$. This is already a
strictly lower representation of the labeled polynomial $cx^\gamma
{\mathcal F}$ w.r.t. set $B_{end}$.
\end{proof}

\begin{prop39}
Let ${\mathcal F}=(x^\alpha\eb_j, f, k_f)\in B_{end}$ be a labeled
polynomial and $cx^\gamma$ a non-zero monomial in $X$. If labeled
polynomial $cx^\gamma{\mathcal F}$ is {\bf rewritable} by
$B_{end}$, then $cx^\gamma{\mathcal F}$ has a strictly lower
representation w.r.t. set $B_{end}$.
\end{prop39}

\begin{proof}
Since $cx^\gamma {\mathcal F}=(x^{\gamma+\alpha}\eb_j, cx^\gamma
f, k_f)$ is rewritable by $B_{end}$, there exists labeled
polynomial ${\mathcal G}=(x^\beta \eb_j, g, k_g)\in B_{end}$ such
that (1) $x^\beta \eb_j \mid x^{\gamma+\alpha}\eb_j$ and (2) $k_f
< k_g$. Denote $x^\lambda=x^{\gamma+\alpha-\beta}$.

On one hand, for labeled polynomial $x^\lambda {\mathcal G}$, since
${\mathcal G}\in B_{end}$, according to Corollary
\ref{cor-signature} (signature),
\begin{equation}\label{eq_19}
\p(x^\lambda {\mathcal G})=x^\lambda g=c_1 x^{\lambda+\beta}
f_j+q_1\p({\mathcal R}_1)+\cdots+q_l\p({\mathcal
R}_l),
\end{equation}
where $c_1$ is a non-zero constant in $\K$, $q_i\in \K[X]$ and
${\mathcal R}_i\in B_{end}$ such that either $q_i=0$ or signature
$\s(x^\lambda {\mathcal G}) \sus \s(q_i{\mathcal R}_i)$ for
$i=1,\cdots,l$. As signature $\s(cx^\gamma {\mathcal F}) =
x^{\gamma+\alpha}\eb_j= x^{\lambda+\beta}\eb_j=
\s(x^\lambda{\mathcal G})$, then signature $\s(cx^\gamma {\mathcal
F}) \sus \s(q_i{\mathcal R}_i)$ and hence labeled polynomial
$cx^\gamma {\mathcal F} \sul q_i{\mathcal R}_i$ for
$i=1,\cdots,l$.

On the other hand, since labeled polynomial ${\mathcal F}\in
B_{end}$ and $cx^\gamma$ is a non-zero monomial, the Corollary
\ref{cor-signature} (signature) shows
\begin{equation}\label{eq_20}
\p(cx^\gamma {\mathcal F})=cx^\gamma f =c_2x^{\gamma+\alpha}
f_j+p_1\p({\mathcal H}_1)+\cdots+p_s\p({\mathcal
H}_s),
\end{equation}
where $c_2$ is a non-zero constant in $\K$, $p_i\in \K[X]$ and
${\mathcal H}_i\in B_{end}$ such that either $p_i=0$ or signature
$\s(cx^\gamma {\mathcal F}) \sus \s(p_i{\mathcal H}_i)$ and hence
labeled polynomial $cx^\gamma {\mathcal F} \sul p_i{\mathcal H}_i$
for $i=1,\cdots,s$.

Since $x^{\lambda+\beta}=x^{\gamma+\alpha}$, substitute the
expression of polynomial $x^{\lambda+\beta}f_j$ in equation
(\ref{eq_19}) into (\ref{eq_20}). Then
\begin{equation}\label{eq_21}
\p(cx^\gamma {\mathcal F})=\frac{c_2}{c_1} (\p(x^\lambda {\mathcal
G})-q_1\p({\mathcal R}_1)-\cdots-q_l\p({\mathcal
R}_l))+p_1\p({\mathcal H}_1)+\cdots+p_s\p({\mathcal
H}_s),
\end{equation}
where $c_1$, $c_2$ are non-zero constants in $\K$, labeled
polynomial $cx^\gamma {\mathcal F}\sul q_i{\mathcal R}_i$ for
$i=1,\cdots,l$ and labeled polynomial $cx^\gamma {\mathcal F} \sul
p_i{\mathcal H}_i$ for $i=1,\cdots,s$. Also notice that labeled
polynomial $cx^\gamma {\mathcal F}\sul x^\lambda{\mathcal G}$, since
signature $\s(cx^\gamma {\mathcal F})= x^{\gamma+\alpha}\eb_j=
x^{\lambda+\beta}\eb_j=\s(x^\lambda{\mathcal G})$ and {\em number}
$\n(cx^\gamma {\mathcal F})=k_f < k_g = \n(x^\lambda{\mathcal G})$.
Then (\ref{eq_21}) is a strictly lower representation of the labeled
polynomial $cx^\gamma {\mathcal F}$ w.r.t. set $B_{end}$.
\end{proof}

The following lemma is the key lemma of the whole proofs, which
shows when a labeled polynomial, who has a {\em strictly lower
representation}, has a {\em $t$-representation}.

\begin{lm310}
Let ${\mathcal F}\in L[X]$ be a labeled polynomial. If
\begin{enumerate}

\item labeled polynomial ${\mathcal F}$ has a strictly lower
representation w.r.t. set $B_{end}$, and

\item all the lower critical pairs of ${\mathcal F}$ have
$t$-representations w.r.t. set $B_{end}$.
\end{enumerate}
Then the labeled polynomial ${\mathcal F}$ has a $t$-representation
w.r.t. set $B_{end}$ where $t=\lpp({\mathcal F})$. Furthermore,
there exists a labeled polynomial ${\mathcal H}\in B_{end}$ such
that: $\lpp({\mathcal H})\mid \lpp({\mathcal F})$ and ${\mathcal
F}\sul x^\lambda{\mathcal H}$ where $x^\lambda=\lpp({\mathcal
F})/\lpp({\mathcal H})$.
\end{lm310}

\begin{proof}
Since labeled polynomial ${\mathcal F}$ has a strictly lower
representation w.r.t. set $B_{end}$, by definition of strictly lower
representation, there exist polynomials $p_1,\cdots,p_s \in \K[X]$
and labeled polynomials ${\mathcal G}_1,\cdots,{\mathcal G}_s\in
B_{end}$, such that: $\p({\mathcal F})=p_1\p({\mathcal
G}_1)+\cdots+p_s\p({\mathcal G}_s),$ where labeled polynomial
${\mathcal F} \sul p_i{\mathcal G}_i$ for $i=1,\cdots,s$.

Let $x^\delta=\max_{\prec}\{\lpp(p_1{\mathcal
G}_1),\cdots,\lpp(p_s{\mathcal G}_s)\}$, so $\lpp({\mathcal
F})\preceq x^\delta$ always holds. Now consider all possible
strictly lower representations of ${\mathcal F}$ w.r.t. set
$B_{end}$. For each such expression, we get a possibly different
$x^\delta$. Since a term order is well-ordering, we can select a
strictly lower representation of ${\mathcal F}$ w.r.t. set
$B_{end}$ such that power product $x^\delta$ is minimal. Assume
this strictly lower representation is
\begin{equation}\label{eq_1}
\p({\mathcal F})=q_1\p({\mathcal H}_1)+\cdots+q_l\p({\mathcal
H}_l),
\end{equation}
where $q_i \in \K[X]$, ${\mathcal H}_i \in B_{end}$ and labeled
polynomial ${\mathcal F} \sul q_i{\mathcal H}_i$ for
$i=1,\cdots,l$. We will show that once this minimal $x^\delta$ is
chosen, we have $\lpp({\mathcal F})=x^\delta$ and hence the lemma
is proved. We prove this by contradiction.

Equality fails only when leading power product $\lpp({\mathcal
F})\prec x^\delta$. Denote $\m(i)=\lpp(q_i{\mathcal H}_i)$, and
then we can rewrite polynomial $\p({\mathcal F})$ in following
form:
$$\p({\mathcal F})=\sum\limits_{\m(i)=x^\delta}q_i\p({\mathcal
H}_i)+\sum\limits_{\m(i)\prec x^\delta}q_i\p({\mathcal H}_i)$$
\begin{equation}\label{eq_3}
=\sum\limits_{\m(i)=x^\delta}\lm(q_i)\p({\mathcal
H}_i)+\sum\limits_{\m(i)=x^\delta}(q_i-\lm(q_i))\p({\mathcal
H}_i)+\sum\limits_{\m(i)\prec x^\delta}q_i\p({\mathcal H}_i).
\end{equation}
The power products appearing in the second and third sums on the
second line all $\prec x^\delta$. Thus, the assumption
$\lpp({\mathcal F})\prec x^\delta$ means that power products in
the first sum also $\prec x^\delta$. So the first sum must be a
linear combination of s-polynomials, i.e.
\begin{equation}\label{eq_2}
\sum\limits_{\m(i)=x^\delta}\lm(q_i)\p({\mathcal
H}_i)=\sum\limits_{j,k}w_{jk}\sp({\mathcal H}_j, {\mathcal H}_k).
\end{equation}
where $w_{jk}$'s are monomials in $X$. For each s-polynomial
$\sp({\mathcal H}_j, {\mathcal H}_k)=u_{jk}{\mathcal
H}_j-v_{jk}{\mathcal H}_k$ in equation (\ref{eq_2}), we have
${\mathcal F}\sul w_{jk}u_{jk}{\mathcal H}_j$, because expression
(\ref{eq_1}) is a strictly lower representation of ${\mathcal F}$.

The next step is to use the hypothesis that all the lower critical
pairs of ${\mathcal F}$ have $t$-representations w.r.t. set
$B_{end}$. Therefore, for each s-polynomial $\sp({\mathcal H}_j,
{\mathcal H}_k)$ in equation (\ref{eq_2}), there exist polynomials
$g_1,\cdots,g_r \in \K[X]$ and labeled polynomials ${\mathcal
R}_1,\cdots,{\mathcal R}_r\in B_{end}$, such that
$$\sp({\mathcal H}_j, {\mathcal H}_k)=g_1\p({\mathcal
R}_1)+\cdots+g_r\p({\mathcal R}_r),$$ where s-polynomial
$\sp({\mathcal H}_j, {\mathcal H}_k) \unrhd g_i{\mathcal R}_i$ and
$\lcm(\lpp({\mathcal H}_j), \lpp({\mathcal H}_k))\succ
\lpp(g_i{\mathcal R}_i)$ for $i=1,\cdots,r$.

Substitute the above representations back into the equation
(\ref{eq_2}) and hence into the equation (\ref{eq_3}). The power
products in the new expression of (\ref{eq_3}) will all $\prec
x^\delta$. Then a new strictly lower representation of ${\mathcal
F}$ w.r.t. set $B_{end}$ appears with all power products $\prec
x^\delta$, which contradicts with the minimality of $x^\delta$. So
we must have $\lpp({\mathcal F})=x^\delta$.

Thus, there exist polynomials $q_1,\cdots,q_l \in \K[X]$ and
labeled polynomials ${\mathcal H}_1,\cdots,{\mathcal H}_l\in
B_{end}$, such that:
\begin{equation}\label{eq_5}
\p({\mathcal F})=q_1\p({\mathcal H}_1)+\cdots+q_l\p({\mathcal
H}_l),
\end{equation}
where ${\mathcal F} \sul q_i{\mathcal H}_i$ and leading power
product $\lpp({\mathcal F}) \succeq \lpp(q_i{\mathcal H}_i)$ for
$i=1,\cdots,l$. And this is already a $t$-representation of
${\mathcal F}$ w.r.t. set $B_{end}$ where $t=\lpp({\mathcal F})$.
Furthermore, since the equality holds in equation (\ref{eq_5}),
there exists an integer $j$ where $1\preceq j \preceq l$, such that
$\lpp({\mathcal F})=\lpp(q_j{\mathcal H}_j)$. The lemma is proved.
\end{proof}

The next two propositions provide sufficient conditions when a
critical pair has a $t$-representation.  Please pay more attention
to the position of ${\mathcal F}$ in the critical pair of each
proposition.

\begin{prop311}
Let $[{\mathcal F}, {\mathcal G}]=(u, {\mathcal F}, v, {\mathcal
G})$ be a critical pair, where ${\mathcal F}, {\mathcal G}\in
B_{end}$ are labeled polynomials and $u, v$ are monomials in $X$
such that $u\lm({\mathcal F})=v\lm({\mathcal
G})=\lcm(\lpp({\mathcal F}), \lpp({\mathcal G}))$. Then the
critical pair $[{\mathcal F}, {\mathcal G}]$ has a
$t$-representation w.r.t. set $B_{end}$, if
\begin{enumerate}
\item labeled polynomial $u{\mathcal F}$ has a strictly lower
representation w.r.t. set $B_{end}$, and

\item all the lower critical pairs of $[{\mathcal F}, {\mathcal
G}]$ have $t$-representations w.r.t. set $B_{end}$.
\end{enumerate}
\end{prop311}

\begin{proof}
Since labeled polynomial $u{\mathcal F}$ has a strictly lower
representation w.r.t. $B_{end}$, then there exist polynomials
$p_1,\cdots,p_s \in \K[X]$ and labeled polynomials ${\mathcal
H}_1,\cdots,{\mathcal H}_s\in B_{end}$, such that $$\p(u{\mathcal
F})=p_1\p({\mathcal H}_1)+\cdots+p_s\p({\mathcal H}_s),$$ where
labeled polynomial $u{\mathcal F} \sul p_i{\mathcal H}_i$ for
$i=1,\cdots,s$. By the definition of critical pairs, labeled
polynomial $u{\mathcal F}\sul v{\mathcal G}$. Then the following
equation holds:
$$\p(\sp({\mathcal F}, {\mathcal G})) = \p(u{\mathcal F}-v{\mathcal G})
= \p(u{\mathcal F})-\p(v{\mathcal G})$$ $$= p_1\p({\mathcal
H}_1)+\cdots+p_s\p({\mathcal H}_s)-v\p({\mathcal G}).$$ Denote
$p_{s+1}=-v$ and ${\mathcal H}_{s+1}={\mathcal G}\in B_{end}$. Then
$$\p(\sp({\mathcal F}, {\mathcal G}))=p_1\p({\mathcal
H}_1)+\cdots+p_s\p({\mathcal H}_s)+p_{s+1}\p({\mathcal H}_{s+1}),$$
where s-polynomial $\sp({\mathcal F}, {\mathcal G}) \bowtie
u{\mathcal F}\sul p_i{\mathcal H}_i$ for $i=1,\cdots,s+1$. Then this
is a strictly lower representation of $\sp({\mathcal F}, {\mathcal
G})$ w.r.t. set $B_{end}$. Combined with the hypothesis that all the
lower critical pairs of $[{\mathcal F}, {\mathcal G}]$ have
$t$-representations w.r.t. set $B_{end}$, Lemma \ref{lem-key} (key
lemma) shows the s-polynomial $\sp({\mathcal F}, {\mathcal G})$ has
a $t$-representation w.r.t. set $B_{end}$ where
$t=\lpp(\sp({\mathcal F}, {\mathcal G}))\prec \lcm(\lpp({\mathcal
F}),\lpp({\mathcal G}))$.
\end{proof}

\begin{prop312}
Let $[{\mathcal G}, {\mathcal F}]=(v, {\mathcal G}, u, {\mathcal
F})$ be a critical pair, where ${\mathcal G}, {\mathcal F}\in
B_{end}$ are labeled polynomials and $v, u$ are monomials in $X$
such that $v\lm({\mathcal G})=u\lm({\mathcal
F})=\lcm(\lpp({\mathcal G}), \lpp({\mathcal F}))$. Then the
critical pair $[{\mathcal G}, {\mathcal F}]$ has a
$t$-representation w.r.t. set $B_{end}$, if
\begin{enumerate}
\item labeled polynomial $u{\mathcal F}$ has a strictly lower
representation w.r.t. set $B_{end}$, and

\item all the lower critical pairs of $[{\mathcal G}, {\mathcal
F}]$ have $t$-representations w.r.t. set $B_{end}$.
\end{enumerate}
\end{prop312}

\begin{proof}
Since labeled polynomial $u{\mathcal F}$ has a strictly lower
representation w.r.t. $B_{end}$ and all the lower critical pairs of
$[{\mathcal G}, {\mathcal F}]$ have $t$-representations w.r.t. set
$B_{end}$, Lemma \ref{lem-key} (key lemma) shows that there exists a
labeled polynomial ${\mathcal H}\in B_{end}$ such that
$\lpp({\mathcal H})\mid \lpp(u{\mathcal F})$ and $u{\mathcal F}\sul
w{\mathcal H}$ where $w=\lm(u{\mathcal F})/\lm({\mathcal H})$.

Notice that $\lpp(v{\mathcal G})=\lpp(u{\mathcal
F})=\lpp(w{\mathcal H})$ and $v{\mathcal G}\sul u{\mathcal F} \sul
w{\mathcal H}$, then
$$\sp({\mathcal G}, {\mathcal F})=v{\mathcal G}-u{\mathcal
F}=(v{\mathcal G}-w{\mathcal H})-(u{\mathcal F}-w{\mathcal H})$$
$$=\gcd(v,w)\sp({\mathcal G},{\mathcal H})-\gcd(u, w)\sp({\mathcal F}, {\mathcal H}).$$
Since critical pair $[{\mathcal G}, {\mathcal F}] \sul [{\mathcal
G}, {\mathcal H}]$ and $[{\mathcal G}, {\mathcal F}] \sul [{\mathcal
F}, {\mathcal H}]$ and all the lower critical pairs of $[{\mathcal
G}, {\mathcal F}]$ have $t$-representations w.r.t. set $B_{end}$,
then the s-polynomial $\sp({\mathcal G}, {\mathcal H})$ has a
$t$-representation w.r.t. set $B_{end}$ where $t\prec
\lcm(\lpp({\mathcal G}),\lpp({\mathcal H}))$, and similarly the
s-polynomial $\sp({\mathcal F},{\mathcal H})$ also has a
$t$-representation w.r.t. set $B_{end}$ where $t\prec
\lcm(\lpp({\mathcal F}), \lpp({\mathcal H}))$.

Combined with the fact that $\lcm(\lpp({\mathcal G}), \lpp({\mathcal
F})) = \gcd(v,w) \lcm(\lpp({\mathcal G}), \lpp({\mathcal H})) =
\gcd(u,w)\lcm(\lpp({\mathcal F}),\lpp({\mathcal H}))$, thus the
s-polynomial $\sp({\mathcal G}, {\mathcal F})$ has a
$t$-representation w.r.t. set $B_{end}$ where $t\prec
\lcm(\lpp({\mathcal G}), \lpp({\mathcal H}))$.
\end{proof}

\section{Available Variation of F5 Algorithm} \label{sec-av}

\subsection{Available Variations}

Briefly, the F5 (F5B) algorithm introduces a special reduction
(F5-reduction) and provides two new criteria (Syzygy Criterion and
Rewritten Criterion) to avoid unnecessary computations/reductions.

From the proofs in last section, Lemma \ref{lem-key} (key lemma)
plays a crucial role in the whole proofs, and the base of this key
Lemma is the property of F5-reduction (Proposition
\ref{prop-reduce}). So {\em the F5-reduction is the key of whole F5
(F5B) algorithm, and it ensures the correctness of the whole
algorithm.}

Therefore, various variations of F5 algorithm become available if
we maintain the F5-reduction. For example,
\begin{enumerate}

\item use various strategies of selecting critical pairs, such as
incremental F5 algorithm in \citep{Fau02} and the F5 algorithm
(reported by Faug\`ere in INSCRYPT 2008);

\item use matrix technique when doing reduction, such as {\em
matrix}-F5 algorithm mentioned in \citep{Bardet03};

\item add new initial polynomials during computation, such as
branch \gr basis algorithm over boolean ring \citep{SunWang09a,
SunWang09b};

\item change the order of signatures, such as \gr basis algorithms
in \citep{Ars09, SunWang09a, SunWang09b}.
\end{enumerate}

Next, we introduce a natural variation of F5 algorithm by change the
order of signatures. This natural variation has been reported in
\citep{SunWang09a, SunWang09b}, and it is also quite similar as the
variation in \citep{Ars09}.

\subsection{A Natural Variation}

In fact, the original F5 algorithm is always an incremental
algorithm no matter which strategy of selecting critical pair is
used. Specifically, the outputs of F5 algorithm not only contain the
\gr basis of the ideal $\langle f_1, \cdots, f_m \rangle$, but also
include the \gr bases of the ideals $\langle f_i, \cdots, f_m
\rangle$ for $1<i<m$.

However, there are three disadvantages of incremental algorithms.
\begin{enumerate}

\item Generally, the ideals $\langle f_i, \cdots, f_m \rangle$ for
$1<i<m$ usually have higher dimensions than the ideal $\langle f_1,
\cdots, f_m \rangle$, so their \gr bases may be expensive to
compute.

\item The \gr bases of ideals $\langle f_i, \cdots, f_m \rangle$
for $1<i<m$ are not necessary, since the \gr of ideal $\langle
f_1, \cdots, f_m \rangle$ is what we really need.

\item The order of initial polynomials influences the efficiency
of algorithm significantly.
\end{enumerate}

If we dig it deeper, we will find that {\em it is the order of
signatures that makes F5 algorithm incremental}. Original F5
algorithm uses a POT (position over term) order of signatures
defined on free module $(\K[X])^m$. Thus, a nature idea is to change
the POT order to the TOP (term over position) order. When using a
TOP order of signatures, F5 algorithm will not be an incremental
algorithm.

We extend the admissible order $\prec$ on $PP(X)$ to free module
$(\K[X])^m$ in the TOP (term over position) fashion:
$$x^\alpha\eb_i \prec' x^\beta\eb_j \mbox{ (or } x^\beta\eb_j \succ' x^\alpha\eb_i)
\ \ \mbox{  iff  } \left\{\begin{array}{l} x^\alpha\lpp(f_i) \prec
x^\beta \lpp(f_j), \\ \mbox{ or }
\\ x^\alpha\lpp(f_i)=x^\beta\lpp(f_j) \mbox{ and } i > j.
\end{array}\right.
$$

Similarly, labeled polynomials are compared in the following way:

$$(x^\alpha\eb_i, f, k_f) \prl' (x^\beta\eb_j, g, k_g)
\mbox{ (or }(x^\beta\eb_j, g, k_g) \sul' (x^\alpha\eb_i, f, k_f))\
\
\mbox{  iff  } \left\{\begin{array}{l} x^\alpha\eb_i \prec' x^\beta\eb_j, \\
\mbox{ or } \\ x^\alpha\eb_i = x^\beta\eb_j \mbox{ and } k_f >
k_g.
\end{array}\right.
$$
Particularly, denote $(x^\alpha\eb_i, f, k_f) \bowtie'
(x^\beta\eb_j, g, k_g)$, if $x^\alpha\eb_i = x^\beta\eb_j$ and
$k_f=k_g$.

There is no need to modify the definition of rewritable, as well
as the descriptions of F5-reduction, Syzygy Criterion and
Rewritten Criterion. However, the definition of comparable needs a
bit adaption to fit the new order.

\begin{define}[new-comparable] \label{df-newcomparable}
Let ${\mathcal F}=(x^\alpha\eb_i, f, k_f)\in L[X]$ be a labeled
polynomial, $cx^\gamma$ a non-zero monomial in $X$ and $B\subset
L[X]$ a set of labeled polynomials. The labeled polynomial
$cx^\gamma{\mathcal F}=(x^{\gamma+\alpha}\eb_i, cx^\gamma f, k_f)$
is said to be new-comparable by $B$, if there exists a labeled
polynomial ${\mathcal G}=(x^\beta\eb_j, g, k_g)\in B$ such that:
\begin{enumerate}

\item $\lpp(g) \mid x^{\gamma+\alpha}$, and

\item $cx^\gamma{\mathcal F} \sul' x^\lambda\lpp(f_i){\mathcal
G}$,  where $x^\lambda = x^{\gamma+\alpha}/\lpp(g)$.
\end{enumerate}
\end{define}

With this definition, the following proposition implies the new
Syzygy Criterion is still correct.

\begin{prop}[new-comparable]
Let ${\mathcal F}=(x^\alpha\eb_j, f, k_f)\in B_{end}$ be a labeled
polynomial and $cx^\gamma$ a non-zero monomial in $X$. If labeled
polynomial $cx^\gamma{\mathcal F}$ is {\bf new-comparable} by
$B_{end}$, then $cx^\gamma{\mathcal F}$ has a strictly lower
representation w.r.t. set $B_{end}$.
\end{prop}

\begin{proof}
As $cx^\gamma {\mathcal F}=(x^{\gamma+\alpha}\eb_j, cx^\gamma f,
k_f)$ is new-comparable by $B_{end}$, there exists labeled
polynomial ${\mathcal G}=(x^\beta\eb_l, g, k_g)\in B_{end}$ such
that (1) $\lpp(g)\mid x^{\gamma+\alpha}$ and (2) $cx^\gamma{\mathcal
F} \sul' x^\lambda\lpp(f_i){\mathcal G}$, where $x^\lambda =
x^{\gamma+\alpha}/\lpp(g)$. Let ${\mathcal F}_j=(\eb_j, f_j, j)\in
B_{end}$ be the labeled polynomial of the initial polynomial $f_j$.
Then the polynomial 2-tuple $(g, -f_j)$ is still a principle syzygy
of the 2-tuple vector $(f_j, g)$ in free module $(\K[X])^2$. So
$$gf_j-f_jg = 0 \mbox{ and } \lm(g)f_j=f_jg-(g-\lm(g))f_j.$$ Since
$x^{\gamma+\alpha}=x^\lambda\lpp(g)$, then
\begin{equation}\label{eq_15}
x^{\gamma+\alpha} f_j=x^\lambda \lpp(g)
f_j=\frac{x^\lambda}{\lc(g)} (f_jg-(g-\lm(g))f_j) =
\frac{x^\lambda}{\lc(g)}
f_jg-\frac{x^\lambda}{\lc(g)}(g-\lm(g))f_j$$ $$ = q_1 g+ q_2 f_j=
q_1\p({\mathcal G})+q_2\p({\mathcal F}_j),
\end{equation}
where $q_1=\frac{x^\lambda}{\lc(g)} f_j$ and
$q_2=-\frac{x^\lambda}{\lc(g)} (g-\lm(g))$.

By the definition of new-comparable, labeled polynomial $cx^\gamma
{\mathcal F}\sull x^\lambda\lpp(f_i){\mathcal G} =
\lpp(q_1){\mathcal G}$. Since
$x^{\gamma+\alpha}\lpp(f_j)=x^\lambda\lpp(g)\lpp(f_j) \succ
x^\lambda\lpp(g-\lm(g))\lpp(f_j) = \lpp(q_2)\lpp(f_j)$, then
$cx^\gamma {\mathcal F}\sull q_2{\mathcal F}_j$ holds.

Since labeled polynomial ${\mathcal F}\in B_{end}$ and $cx^\gamma$
is a non-zero monomial, Corollary \ref{cor-signature} (signature)
shows
\begin{equation}\label{eq_16}
\p(cx^\gamma {\mathcal F})=cx^\gamma f=\bar{c}x^{\gamma+\alpha}
f_j+p_1\p({\mathcal H}_1)+\cdots+p_s\p({\mathcal H}_s),
\end{equation}
where $\bar{c}$ is a non-zero constant in  $\K$, $p_i\in \K[X]$
and ${\mathcal H}_i\in B_{end}$ such that either $p_i=0$ or
signature $\s(cx^\gamma{\mathcal F}) \sus' \s(p_i{\mathcal H}_i)$
and hence labeled polynomial $cx^\gamma {\mathcal F} \sul'
p_i{\mathcal H}_i$ for $i=1,\cdots,s$.

Substitute the expression of polynomial $x^{\gamma+\alpha} f_j$ in
equation (\ref{eq_15}) into (\ref{eq_16}). Then
$$\p(cx^\gamma {\mathcal F})=\bar{c} q_1\p({\mathcal
G})+\bar{c} q_2\p({\mathcal F}_j)+p_1\p({\mathcal
H}_1)+\cdots+p_s\p({\mathcal H}_s),$$ where labeled polynomial
$cx^\gamma {\mathcal F} \sul' \bar{c}q_1{\mathcal G}$, $cx^\gamma
{\mathcal F} \sul' \bar{c}q_2{\mathcal F}_j$ and $cx^\gamma
{\mathcal F} \sul' p_i{\mathcal H}_i$ for $i=1,\cdots,s$. This is
already a strictly lower representation of the labeled polynomial
$cx^\gamma {\mathcal F}$ w.r.t. set $B_{end}$.
\end{proof}

\begin{remark}
For the labeled polynomials ${\mathcal F}=(x^\alpha\eb_i, f, k_f)\in
L[X]$ and ${\mathcal G}=(x^\beta\eb_j, g, k_g)\in B$ in Definition
\ref{df-newcomparable} (new-comparable). The second condition
``$cx^\gamma{\mathcal F} \sul' x^\lambda\lpp(f_i){\mathcal G}$" is
in fact equivalent to the condition ``signature $\s({\mathcal
G})=\eb_j$ and $\eb_i\succ \eb_j$, i.e. $i < j$". So this new Syzygy
Criterion only utilizes the principle syzygies of initial
polynomials, which is the same as the criteria in \citep{Ars09}. The
technique ``adding new initial polynomials during computation"
introduced in \citep{SunWang09a, SunWang09b} will enhance this new
Syzygy Criterion. Specifically, when a labeled polynomial ${\mathcal
P}=(x^\gamma\eb_l, p, k_p)$ is generated during the computation,
simply adding the labeled polynomial ${\mathcal P}'=(\eb_{l'}, p,
k'_p)$ into computation and updating critical pairs correspondingly
do not affect the correctness of algorithm, where we prefer $l'>l$
and $k'_p>k_p$ such that ${\mathcal P}\sul {\mathcal P}'$.
\end{remark}

The Syzygy Criterion in Ars and Hashema's paper \citep{Ars09} can
also be proved in a similar way as above.

\subsection{Criteria of the Natural Variation}

Although only the principle syzygies of initial polynomials are
used, the new Syzygy Criterion also performs pretty good in
experiments. We have implemented this natural variation of F5
algorithm over boolean ring \citep{SunWang09a, SunWang09b}. The data
structure ZDD (Zero-suppressed Binary Decision Diagrams) is used to
express boolean polynomials, and the ``adding new initial
polynomials during computation" technique is also used to enhance
the new Syzygy Criterion. Also matrix technique is used when
F5-reducing labeled polynomials, but this procedure is not fully
optimized yet, as only general Gaussian elimination is used.

The data about the two revised criteria in following table are
obtained from the above implementation. Examples are randomly
generated quadratic boolean polynomials, and the number of initial
polynomials $m$ equals to the number of variables $n$. The timings
are obtained from a computer (OS Linux, CPU Xion 4*3.0GHz, 16.0GB
RAM). In the table \ref{tab-criteria}, {\em comparable,
F2-comparable} \footnote{F2-comparable is a special {\em comparable}
which results from the characteristic of boolean ring, since for
each boolean polynomial $f$, we always have $f^2=f$ in boolean ring.
For more details, please see \citep{SunWang09a, SunWang09b}.} and
{\em rewritable} refer to the times of corresponding conditions
being met in the computation. Remark that these numbers are not the
numbers of rejected critical pairs, as F5-reduction also needs to
check the {\em comparable, F2-comparable} and {\em rewritable}.
Besides, {\em useful cp's} is the number of critical pairs that are
really operated during computation (i.e. not rejected by two
criteria). {\em 0-polys} is the number of labeled polynomials that
F5-reduce to $0$.

\begin{table}[!ht]\centering
\caption{The Revised Criteria} \label{tab-criteria}
\begin{tabular}{ccccccccc}
\hline\noalign{\smallskip}
$m=n$ & 6 & 8 & 10 & 12 & 14 & 16 & 18 & 20  \\
\noalign{\smallskip}\hline\noalign{\smallskip}
   comparable    & 0 & 36 & 566 & 898 & 72189 & 68337 & 99058 & 136404 \\
   F2-comparable & 2 & 5 & 87 & 114 & 7770 & 6763 & 9374 & 11749 \\
   rewritable    & 2 & 20 & 74 & 136 & 6908 & 4786 & 6293 & 8536 \\
   useful cp's      & 21 & 77 & 225 & 305 & 841 & 3480 & 4469 & 5672 \\
   0-polys       & 0 & 0 & 0 & 0 & 0 & 0 & 0 & 0 \\
   Time(sec.)    & 0.001 & 0.005 & 0.034 & 0.107 & 0.778 & 14.586 & 77.197 & 344.875\\
\noalign{\smallskip}\hline
\end{tabular}
\end{table}

From the data in table \ref{tab-criteria}, most of redundant
computations/reductions are rejected by the revised {\em
new-comparable} (Syzygy Criterion), particularly in large examples
and no labeled polynomials F5-reduce to $0$ in these examples.
Therefore, the revised criteria in the natural variation of F5
algorithm are very effective and they are able to reject almost all
unnecessary computations/reductions.

\section{Conclusion}

In this paper, a complete proof for the correctness of F5 (F5-like)
algorithm is presented. As F5B algorithm is equivalent to the
original F5 algorithm as well as some F5-like algorithms, we
concentrate on the proof for the correctness of F5B algorithm. This
new proposed proof is not limited to homogeneous systems and does
not depend on the strategies of selecting critical pairs, so it can
easily extends to other variations of F5 algorithm. From the new
proof, we find that the F5-reduction is the key of the whole
algorithm and it ensures the correctness of two criteria. With these
insights, various variations of F5 algorithm become available by
maintaining the F5-reduction. We present and prove a natural
variation of F5 algorithm which is not incremental. We hope to study
other variations of F5 algorithm in the future.

\section{Acknowledgements}

We would like to thank Professor Xiaoshan Gao, Professor Deepak
Kapur, Professor Shuhong Gao and Christian Eder for their
constructive suggestions.

\end{document}